\documentclass[journal,twocolumn,10pt,twoside]{IEEEtranTCOM}
\newtheorem{definition}{Definition}
\newtheorem{theorem}{Theorem}
\newtheorem{lemma}{Lemma}
\newtheorem{remark}{Remark}

\usepackage{cite}
\usepackage{amsmath,amssymb,amsfonts}
\usepackage{algorithmic}
\usepackage{graphicx}
\usepackage{textcomp}
\usepackage{xcolor}
\usepackage{subcaption}  
\usepackage{booktabs}   
\usepackage{hyperref}
\allowdisplaybreaks

\normalsize

\ifCLASSINFOpdf
\else
\fi

\begin{document}
		
		\title{On the Fundamental Limits of Bistatic ISAC Over Vector Channels}

			
				\author{Dongsheng Peng,
				Lei Yu,
				Ping Chen,
				and Jun Chen
			
			
			
		}
		
		\maketitle

		\begin{abstract}
			This paper studies the fundamental limits of bistatic integrated sensing and communications (ISAC) systems over vector channels. We first consider a class of parallel channels in which, on each component channel, the channel from the transmitter (Tx) to the communication receiver (Rx) is degraded with respect to the channel from the Tx to the sensing Rx, or vice versa, where the direction of degradation may vary across component channels. For this class of channels, we establish matching achievability and converse results, which yield a single-letter characterization of the capacity-distortion function $C(D)$. For parallel Gaussian channels with additive Gaussian states, the evaluation of the capacity-distortion-power function $C(D,P)$ under the squared error distortion measure reduces to a joint power and distortion allocation problem across the component channels. We then extend the analysis to general vector channels and derive single-letter lower and upper bounds on $C(D)$. In particular, for vector Gaussian channels with additive Gaussian states, we prove the optimality of Gaussian superposition coding without requiring a global degradation order between the communication and sensing channels, by combining the channel enhancement technique with a matrix extremal inequality based on Fisher information. This leads to an exact characterization of $C(D,P)$.
		\end{abstract}


		%
		\IEEEpeerreviewmaketitle

		\section{Introduction}
		
		Future wireless networks are expected to reuse the spectrum, waveforms, and hardware of communication systems to sense the surrounding environment~\cite{Liu2022JSAC,LiuA2022Survey}. This technology, which simultaneously performs information transmission and environmental sensing over shared wireless resources, is referred to as integrated sensing and communications (ISAC). A central question in ISAC system design is how much communication rate must be sacrificed to achieve a given sensing accuracy, and which transmission strategy achieves the optimal tradeoff. Information theory provides a natural performance benchmark for this question: the sensing target is modeled as a random state of the channel, the sensing accuracy is measured by the expected distortion of state estimation, and the fundamental tradeoff between communication rate and sensing accuracy is characterized by the capacity-distortion function $C(D)$, which is defined as the maximum achievable communication rate subject to a distortion constraint $D$.
		
		Existing information-theoretic studies have mainly focused on monostatic sensing, in which the sensing receiver (Rx) is colocated with the transmitter (Tx) and therefore has knowledge of the channel input. For memoryless state-dependent channels with independent and identically distributed (i.i.d.) states, the tradeoff between communication and sensing in this case is governed solely by constraints on the channel input distribution~\cite{Kobayashi2018ISIT,Ahmadipour2024IT}. In Gaussian channels, this structure manifests itself as a deterministic-random tradeoff: sensing favors deterministic transmit signals, whereas the communication rate relies on the randomness of the signal~\cite{Xiong2023IT,LiuF2023ISIT}. Closely related works also include early studies on joint transmission and state estimation~\cite{Zhang2011IT}, channel capacity and state estimation in state-dependent Gaussian channels~\cite{Sutivong2005IT,Kim2008IT}, and the formulation of the monostatic ISAC problem under logarithmic loss in~\cite{Joudeh2024JCS}.

		In bistatic sensing, the sensing Rx is physically separated from the Tx, which avoids full-duplex self-interference and can provide geometric diversity~\cite{GriffithsBaker2017,Colone2009TAES}. From an information-theoretic perspective, the sensing Rx knows the codebook but not the transmitted message; hence, the randomness of the information-bearing signal acts as a form of interference for sensing. Jiao et al.~\cite{Jiao2023bistatic} established a capacity-distortion framework for bistatic ISAC under the assumption that the communication Rx knows the channel state, derived a multi-letter characterization of $C(D)$ together with several single-letter lower and upper bounds, and showed that these bounds are tight for degraded channels.
		Subsequently, Chen et al.~\cite{Chen2026TCOM} relaxed the assumption that the communication Rx knows the channel state and adopted logarithmic loss to measure the sensing accuracy, requiring the sensing Rx to output a soft estimate of the state sequence. They showed that $C(D)$ admits a single-letter characterization whenever one of the communication and sensing channels is degraded with respect to the other. Specifically, when the communication channel is degraded with respect to the sensing channel, $C(D)$ exhibits a decoupling property similar to that in monostatic sensing, i.e., the sensing distortion affects the communication rate only through constraints on the channel input distribution; when the sensing channel is degraded with respect to the communication channel, a superposition coding scheme with partial decoding at the sensing Rx achieves the optimal tradeoff.
		In addition, Jiao et al.~\cite{Jiao2026AB} transformed the nonconvex distortion constraint into a linear one by introducing auxiliary variables, and proposed an extended Arimoto--Blahut algorithm applicable to both squared error and logarithmic loss for numerically computing $C(D)$ in bistatic ISAC.

		However, existing exact single-letter characterizations of $C(D)$ for bistatic ISAC mainly rely on a global degradation order between the communication and sensing channels, a condition that is often difficult to satisfy in multicarrier and multi-antenna systems. In orthogonal frequency-division multiplexing (OFDM) ISAC systems~\cite{Sturm2011ProcIEEE}, the communication Rx and the sensing Rx are located at different positions and experience different frequency-selective fading; hence, their relative strength may vary across subcarriers, and the direction of degradation may differ from one subcarrier to another. In multi-antenna systems, the effective noise covariance matrices of the communication and sensing channels are in general not ordered in the positive semidefinite sense, and thus no global degradation order exists between the two channels. In classical broadcast channel theory, these two cases correspond to the product of reversely degraded broadcast channels~\cite{ElGamal1980product} and the Gaussian multiple-input multiple-output (MIMO) broadcast channel~\cite{Weingarten2006IT}, respectively. In ISAC, however, the task of the sensing Rx is to estimate the channel state subject to a distortion constraint rather than to decode its own message. Since the sensing distortion is coupled with both the channel input distribution and the partial decoding at the sensing Rx, the above classical results do not directly apply.
		
		To address these two cases, this paper extends the analysis of bistatic ISAC to parallel channels and general vector channels. For parallel channels, different component channels are allowed to have different directions of degradation; for general vector channels, a global degradation order between the communication and sensing channels is no longer required. The main contributions of this paper are summarized as follows.
		\begin{itemize}
			\item \emph{Parallel channels:} For parallel channels in which each component channel is degraded in one direction but the direction of degradation may differ across component channels, we establish matching achievability and converse results, which yield a single-letter characterization of $C(D)$. In the achievable scheme, the common layer is jointly encoded across all component channels, which exploits the complementary degradation directions of the two types of component channels and yields a rate gain over independent coding on each component channel. Furthermore, for parallel Gaussian channels with additive Gaussian states under the squared error distortion measure, building on the optimality of Gaussian superposition inputs for scalar Gaussian channels in~\cite{Chen2026TCOM}, we reduce the computation of the capacity-distortion-power function $C(D,P)$ to a convex joint power and distortion allocation problem across the component channels, and show that its solution has a water-filling structure in both power and distortion.
			
			\item \emph{General vector channels:} For general vector channels without a global degradation order between the communication and sensing channels, we establish single-letter lower and upper bounds on $C(D)$. Furthermore, for vector Gaussian channels with additive Gaussian states, we use the conditional Fisher information matrix of the sensing observation to convert the squared error distortion constraint into a constraint on an equivalent covariance matrix. Building on this, we construct an enhanced communication channel based on the Karush--Kuhn--Tucker (KKT) conditions of the corresponding optimization problem and, combined with a matrix extremal inequality based on Fisher information, prove that Gaussian superposition coding is optimal, thereby yielding an exact characterization of $C(D,P)$. In this optimal scheme, the sensing Rx performs partial decoding to assist state estimation.
		\end{itemize}

		The rest of this paper is organized as follows. Section~\ref{sec:model} introduces the bistatic ISAC system model and defines $C(D)$. Section~\ref{sec:parallel} presents a single-letter characterization of $C(D)$ for parallel channels and discusses the computation of $C(D,P)$ for parallel Gaussian channels. Section~\ref{sec:vector} derives single-letter lower and upper bounds on $C(D)$ for general vector channels and provides an exact characterization of $C(D,P)$ for vector Gaussian channels with additive Gaussian states. Section~\ref{sec:conclusion} concludes the paper.
		
		\emph{Notation}: Scalar random variables and their realizations are denoted by uppercase and lowercase italic letters, respectively, e.g., $X$ and $x$; random vectors and their realizations are denoted by boldface uppercase and lowercase letters, respectively, e.g., $\mathbf{X}$ and $\mathbf{x}$; matrices are denoted by non-boldface uppercase letters, with the exception of the identity matrix and the zero matrix, which are denoted by $\mathbf{I}$ and $\mathbf{0}$, respectively. Sets are denoted by calligraphic letters, $|\mathcal{X}|$ denotes the cardinality of $\mathcal{X}$, and $[a:b]:=\{a,\ldots,b\}$. We write $X^n:=(X_1,\ldots,X_n)$, $X_{t+1}^n:=(X_{t+1},\ldots,X_n)$, and $X^{n\setminus t}:=(X^{t-1},X_{t+1}^n)$. The logarithm $\log$ is to base $2$, and $\ln$ denotes the natural logarithm; $J(\cdot)$ denotes the Fisher information matrix defined with respect to the natural logarithm. For matrices $A$ and $B$, $A\succeq B$ ($A\succ B$) means that $A-B$ is positive semidefinite (positive definite); $|A|$ and $\operatorname{tr}(A)$ denote the determinant and the trace of $A$, respectively; and $\|\cdot\|_2$ denotes the Euclidean norm of a vector or the spectral norm of a matrix. Finally, $\operatorname{var}(X|Y) := \mathbb{E}[(X - \mathbb{E}[X|Y])^2]$ and $\operatorname{Cov}(\mathbf{X}|\mathbf{Y}) := \mathbb{E}[(\mathbf{X} - \mathbb{E}[\mathbf{X}|\mathbf{Y}])(\mathbf{X} - \mathbb{E}[\mathbf{X}|\mathbf{Y}])^{\top}]$ denote the conditional variance and the conditional covariance matrix, respectively.

		\section{System Model}\label{sec:model}
		\begin{figure}[t]
			\centering
			\includegraphics[width=1.0\linewidth]{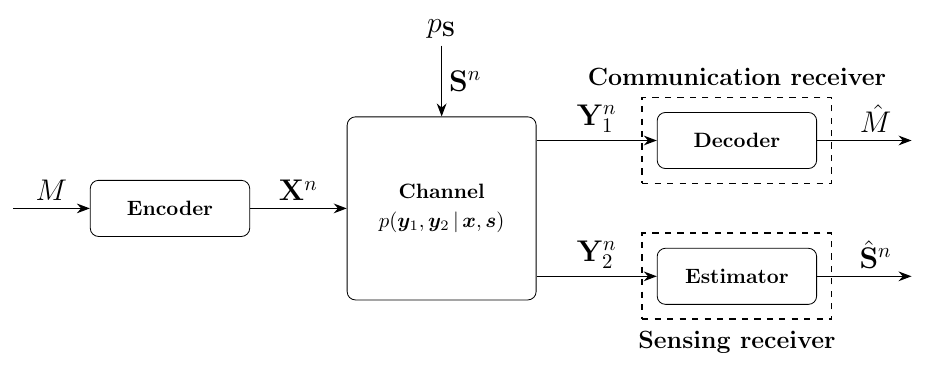}
			\caption{Bistatic ISAC system over a vector channel.}
			\label{fig:model}
		\end{figure}
		Consider the bistatic ISAC system shown in Fig.~\ref{fig:model}, which consists of a Tx, a communication Rx, and a sensing Rx. The Tx encodes a message $M$ into a codeword $\mathbf{X}^n$ of length $n$ and transmits it over a state-dependent memoryless broadcast channel. The communication Rx recovers the message $M$ from its received sequence $\mathbf{Y}_1^n$, and the sensing Rx estimates the channel state sequence $\mathbf{S}^n$ from its received sequence $\mathbf{Y}_2^n$.

		In each channel use, the channel input, the output at the communication Rx, the output at the sensing Rx, and the channel state are denoted by $\mathbf{X}\in\mathcal{X}$, $\mathbf{Y}_1\in\mathcal{Y}_1$, $\mathbf{Y}_2\in\mathcal{Y}_2$, and $\mathbf{S}\in\mathcal{S}$, respectively. The state sequence $\mathbf{S}^n=(\mathbf{S}_1,\ldots,\mathbf{S}_n)$ is i.i.d. according to $p_{\mathbf{S}}$ and is independent of the message $M$. The channel is memoryless, i.e.,
		\begin{align}
			&p_{\mathbf{Y}_1^n\mathbf{Y}_2^n\mathbf{S}^n|\mathbf{X}^n}(\mathbf{y}_1^n,\mathbf{y}_2^n,\mathbf{s}^n|\mathbf{x}^n)\notag\\
			&\quad=\prod_{t=1}^{n}p_{\mathbf{S}}(\mathbf{s}_t)\,
			p_{\mathbf{Y}_1\mathbf{Y}_2|\mathbf{X}\mathbf{S}}(\mathbf{y}_{1,t},\mathbf{y}_{2,t}|\mathbf{x}_t,\mathbf{s}_t),\label{eq:model-memoryless}
		\end{align}
		where the subscript $t$ indicates the $t$th channel use. For general discrete channels, all alphabets are assumed to be finite; for the Gaussian channels considered later, the random vectors take values in the corresponding Euclidean spaces.
		
		Let the message $M$ be uniformly distributed over $\mathcal{M}$. An $(n,\mathcal{M})$ code consists of the following three components:
		\begin{itemize}
			\item an encoder $f^{(n)}:\mathcal{M}\to\mathcal{X}^n$ that generates the codeword $\mathbf{X}^n=f^{(n)}(M)$;
			\item a communication decoder $g_1^{(n)}:\mathcal{Y}_1^n\to\mathcal{M}$ that outputs $\hat{M}=g_1^{(n)}(\mathbf{Y}_1^n)$; and
			\item a sensing estimator $g_2^{(n)}:\mathcal{Y}_2^n\to\hat{\mathcal{S}}^n$ that outputs $\hat{\mathbf{S}}^n=g_2^{(n)}(\mathbf{Y}_2^n)$, where $\hat{\mathcal{S}}$ is the state reconstruction alphabet.
		\end{itemize}
		
		The sensing performance is measured by a per-letter distortion function $d:\mathcal{S}\times\hat{\mathcal{S}}\to[0,\infty)$, and the average distortion is defined as
		\begin{equation}\label{eq:model-distortion}
			D^{(n)}
			:=\frac{1}{n}\sum_{t=1}^{n}
			\mathbb{E}\left[
			d(\mathbf{S}_t,\hat{\mathbf{S}}_t)
			\right].
		\end{equation}
		For discrete channels, since the alphabets are finite, we have $d_{\max}:=\max_{\mathbf{s},\hat{\mathbf{s}}}d(\mathbf{s},\hat{\mathbf{s}})<\infty$. Moreover, for any random variable $\boldsymbol{\Theta}$, we define the minimum conditional expected distortion as
		\begin{equation}\label{eq:model-dstar}
			d^*(\mathbf{S}|\boldsymbol{\Theta}):=\min_{\hat{s}(\cdot)}\mathbb{E}\big[d(\mathbf{S},\hat{s}(\boldsymbol{\Theta}))\big],
		\end{equation}
		where the minimization is over all mappings $\hat{s}(\cdot)$, and the minimum is attained by the Bayes estimator
		\begin{equation}\label{eq:model-bayes}
			\hat{s}^*(\boldsymbol{\theta}) \in \arg\min_{\hat{\mathbf{s}}\in\hat{\mathcal{S}}}\mathbb{E}\big[d(\mathbf{S},\hat{\mathbf{s}})\,\big|\,\boldsymbol{\Theta}=\boldsymbol{\theta}\big].
		\end{equation}
		
		\begin{definition}\label{def:achievable}
			For a given distortion constraint $D$, a rate-distortion pair $(R,D)$ is said to be achievable if, for any $\epsilon>0$ and all sufficiently large $n$, there exists an $(n,\mathcal{M})$ code such that
			\begin{align}
				\frac{1}{n}\log|\mathcal{M}|
				&\geq R-\epsilon, \label{eq:def-rate}\\
				\mathbb{P}\{\hat{M}\neq M\}
				&\leq\epsilon, \label{eq:def-error}\\
				D^{(n)}
				&\leq D+\epsilon. \label{eq:def-distortion}
			\end{align}
			The capacity-distortion function is defined as
			\begin{equation}\label{eq:model-CD}
				C(D)=\sup\{R:(R,D)\text{ is achievable}\}.
			\end{equation}
		\end{definition}

		\section{Parallel Channels}\label{sec:parallel}
		
		In this section, we consider a parallel bistatic ISAC system consisting of $K$ independent component channels. Let $\mathbf{X}=(X_1,\ldots,X_K)$, $\mathbf{Y}_i=(Y_{i1},\ldots,Y_{iK})$ ($i=1,2$), and $\mathbf{S}=(S_1,\ldots,S_K)$ denote the channel input, the channel outputs, and the state, respectively, where the subscript $k$ corresponds to the $k$th component channel; the corresponding components in the $t$th channel use are denoted by $X_{k,t}$, $Y_{ik,t}$, and $S_{k,t}$. The channel has a product structure, i.e.,
		\begin{equation}\label{eq:par-channel}
			p_{\mathbf{Y}_1\mathbf{Y}_2\mathbf{S}|\mathbf{X}}
			=\prod_{k=1}^{K}p_{S_k}\,p_{Y_{1k}Y_{2k}|X_kS_k},
		\end{equation}
		which means that the states of different component channels are mutually independent, and the outputs of the $k$th component channel depend only on $(X_k,S_k)$. Furthermore, we assume that $\hat{\mathcal{S}}=\hat{\mathcal{S}}_1\times\cdots\times\hat{\mathcal{S}}_K$ and that the distortion function is additive across component channels, i.e., $d(\mathbf{s},\hat{\mathbf{s}})=\sum_{k=1}^{K}d_k(s_k,\hat{s}_k)$,
		where $d_k:\mathcal{S}_k\times\hat{\mathcal{S}}_k\to[0,\infty)$ is the distortion function of the $k$th component channel.
		
		We do not require all component channels to have the same direction of degradation. Specifically, suppose that $[1:K]$ is partitioned into two disjoint sets $\mathcal{A}$ and $\mathcal{B}$ such that
		\begin{subequations}\label{eq:par-degraded}
			\begin{align}
				X_k&\leftrightarrow Y_{2k}\leftrightarrow Y_{1k},
				\qquad k\in\mathcal{A},\label{eq:par-degraded-a}\\
				X_k&\leftrightarrow Y_{1k}\leftrightarrow Y_{2k},
				\qquad  k\in\mathcal{B},\label{eq:par-degraded-b}
			\end{align}
		\end{subequations}
		that is, for each component channel in $\mathcal{A}$, the communication channel is degraded with respect to the sensing channel, whereas the opposite holds for each component channel in $\mathcal{B}$. Without loss of generality, let $\mathcal{A}=[1:r]$ and $\mathcal{B}=[r+1:K]$, where $0\le r\le K$; that is, either $\mathcal{A}$ or $\mathcal{B}$ may be empty.
		\begin{remark}\label{rem:marginal}
			Since the communication decoder uses only $\mathbf{Y}_1^n$ and the average distortion~\eqref{eq:model-distortion} depends only on $(\mathbf{S}^n,\mathbf{Y}_2^n)$, $C(D)$ depends on $p_{\mathbf{Y}_1\mathbf{Y}_2\mathbf{S}|\mathbf{X}}$ only through $p_{\mathbf{Y}_1|\mathbf{X}}$ and $p_{\mathbf{Y}_2\mathbf{S}|\mathbf{X}}$~\cite{Chen2026TCOM}. Hence, when deriving upper bounds, one may choose any joint distribution that preserves these two conditional distributions. In particular, the physical degradedness in~\eqref{eq:par-degraded} can be relaxed to stochastic degradedness~\cite[Sec.~III-C]{Chen2026TCOM}.
		\end{remark}
		
		The following theorem provides a single-letter characterization of $C(D)$ for parallel bistatic ISAC channels satisfying~\eqref{eq:par-degraded}.
		\begin{theorem}\label{thm:parallel}
			For the parallel bistatic ISAC channel satisfying~\eqref{eq:par-degraded},
			\begin{subequations}\label{eq:par-thm-opt}
				\begin{align}
					C(D)=\max\;&\min\{R_1,R_2\}\label{eq:par-thm}\\
					\text{s.t.}\;&\sum_{k\in\mathcal{A}}d_k^*(S_k|X_k,Y_{2k})+\sum_{k\in\mathcal{B}}d_k^*(S_k|U_k,Y_{2k})\le D,\label{eq:par-thm-dist}
				\end{align}
			\end{subequations}
			where
			\begin{align*}
				R_1&:=\sum_{k=1}^{K}I(X_k;Y_{1k}),\\
				R_2&:=\sum_{k\in\mathcal{A}}I(X_k;Y_{2k})
				+\sum_{k\in\mathcal{B}}\big[I(X_k;Y_{1k}|U_k)+I(U_k;Y_{2k})\big],
			\end{align*}
			the maximum is over all distributions $\prod_{k\in\mathcal{A}}p_{X_k}\prod_{k\in\mathcal{B}}p_{U_kX_k}$, the joint distribution factorizes as $\prod_{k\in\mathcal{A}}p_{X_k}\prod_{k\in\mathcal{B}}p_{U_kX_k}\,p_{\mathbf{Y}_1\mathbf{Y}_2\mathbf{S}|\mathbf{X}}$, the auxiliary random variables satisfy $|\mathcal{U}_k|\le|\mathcal{X}_k|+1$ for $k\in\mathcal{B}$, and $d_k^*(\cdot|\cdot)$ is defined as in~\eqref{eq:model-dstar}.
		\end{theorem}
		\begin{remark}\label{rem:par-special}
			Theorem~\ref{thm:parallel} includes as special cases the two scenarios in which all component channels have the same direction of degradation. If $\mathcal{B}=\emptyset$, then by~\eqref{eq:par-degraded-a} and the data processing inequality, $I(X_k;Y_{1k})\le I(X_k;Y_{2k})$; hence, the minimum in~\eqref{eq:par-thm} is attained by the first term, and thus
			\begin{equation*}
				C(D)=\max_{\prod_{k=1}^{K}p_{X_k}}\Bigg\{\sum_{k=1}^{K}I(X_k;Y_{1k}):\sum_{k=1}^{K}d_k^*(S_k|X_k,Y_{2k})\le D\Bigg\}.
			\end{equation*}
			In this case, the communication rate and the sensing distortion are coupled only through the input distribution, which is consistent with the decoupling structure in~\cite[Corollary~1]{Chen2026TCOM}. If $\mathcal{A}=\emptyset$, then by~\eqref{eq:par-degraded-b}, $I(U_k;Y_{2k})\le I(U_k;Y_{1k})$; hence, the minimum is attained by the second term, which yields the superposition coding structure corresponding to~\cite[Corollary~2]{Chen2026TCOM}.
			
			In the mixed-degradation case, on each component channel in $\mathcal{A}$, the sensing Rx can obtain the corresponding channel input $X_k$ through full decoding, and the resulting minimum expected distortion is $d_k^*(S_k|X_k,Y_{2k})$; on each component channel in $\mathcal{B}$, the sensing Rx obtains the auxiliary random variable $U_k$ through partial decoding, and the resulting minimum expected distortion is $d_k^*(S_k|U_k,Y_{2k})$. The communication rates over the two types of component channels are jointly constrained by the two sum mutual information terms in~\eqref{eq:par-thm}.
		\end{remark}
		\begin{remark}\label{rem:par-gain}
			Consider a scheme that encodes separately on each component channel, using single-layer coding on the component channels in $\mathcal{A}$ and superposition coding on those in $\mathcal{B}$. For a given product distribution, this scheme achieves the same distortion, and its achievable rate is
			\begin{equation*}
				R_{\mathrm{sep}}
				=\sum_{k\in\mathcal A}I(X_k;Y_{1k})
				+\sum_{k\in\mathcal B}
				\big[I(X_k;Y_{1k}|U_k)+I(U_k;Y_{2k})\big].
			\end{equation*}
			In contrast, the objective function in~\eqref{eq:par-thm} can be written as
			\begin{equation*}
				R_{\mathrm{joint}}
				=R_{\mathrm{sep}}+\min\{\Delta_{\mathcal A},\Delta_{\mathcal B}\},
			\end{equation*}
			where
			\begin{align*}
				\Delta_{\mathcal A}
				&:=\sum_{k\in\mathcal A}
				\big[I(X_k;Y_{2k})-I(X_k;Y_{1k})\big],\\
				\Delta_{\mathcal B}
				&:=\sum_{k\in\mathcal B}
				\big[I(U_k;Y_{1k})-I(U_k;Y_{2k})\big].
			\end{align*}
			By~\eqref{eq:par-degraded} and the data processing inequality, $\Delta_{\mathcal{A}},\Delta_{\mathcal{B}}\ge0$; therefore, $R_{\mathrm{joint}}\ge R_{\mathrm{sep}}$, with strict inequality whenever both $\Delta_{\mathcal{A}}>0$ and $\Delta_{\mathcal{B}}>0$. Intuitively, under separate coding, the surplus decoding capability $\Delta_{\mathcal{A}}$ of the sensing Rx on the component channels in $\mathcal{A}$ is left unused, whereas the cloud-layer rate on the component channels in $\mathcal{B}$ is limited by the weaker sensing channel, which incurs a loss of $\Delta_{\mathcal{B}}$. By jointly encoding the common layer across all component channels, the surplus of the former can compensate for the loss of the latter.
		\end{remark}

		\subsection{Achievability}\label{subsec:par-ach}
		
		The achievability proof is based on superposition coding~\cite[Sec.~5.3]{ElGamalKim2011}, in which the common layer is jointly encoded across all component channels. Fix a distribution $\prod_{k\in\mathcal{A}}p_{X_k}\prod_{k\in\mathcal{B}}p_{U_kX_k}$. For notational convenience, let $U_k:=X_k$ for $k\in\mathcal{A}$, and denote the common layer by $\mathbf{U}:=(U_1,\ldots,U_K)$; we also write $\mathbf{X}_{\mathcal{B}}:=(X_k)_{k\in\mathcal{B}}$ and $\mathbf{Y}_{1\mathcal{B}}:=(Y_{1k})_{k\in\mathcal{B}}$.
		
		\textit{Codebook generation:} Split the message $M$ into independent common and private messages $M_0\in[1:2^{nR_u}]$ and $M_1\in[1:2^{nR_v}]$, where $R=R_u+R_v$. For each $m_0$, independently generate a common-layer codeword $\mathbf{u}^n(m_0)$ according to $\prod_{t=1}^{n}p_{\mathbf{U}}(\mathbf{u}_t)$, where $p_{\mathbf{U}}=\prod_{k=1}^{K}p_{U_k}$. For each $(m_0,m_1)$, conditionally independently generate a private-layer codeword $\mathbf{x}_{\mathcal{B}}^n(m_0,m_1)$ according to $\prod_{t=1}^{n}\prod_{k\in\mathcal{B}}p_{X_k|U_k}\big(x_{k,t}|u_{k,t}(m_0)\big)$.
		
		\textit{Encoding:} To send $(m_0,m_1)$, the Tx transmits $\mathbf{u}_{\mathcal{A}}^n(m_0)$ over the component channels in $\mathcal{A}$ and $\mathbf{x}_{\mathcal{B}}^n(m_0,m_1)$ over those in $\mathcal{B}$.
		
		\textit{Decoding:} The sensing Rx finds the unique $\hat{m}_0$ such that $\big(\mathbf{u}^n(\hat{m}_0),\mathbf{y}_2^n\big)\in\mathcal{T}_\epsilon^{(n)}$. The communication Rx performs successive decoding: it first decodes $\hat{m}_0$ from $\mathbf{y}_1^n$ in the same manner, and then finds the unique $\hat{m}_1$ such that $\big(\mathbf{u}_{\mathcal{B}}^n(\hat{m}_0),\mathbf{x}_{\mathcal{B}}^n(\hat{m}_0,\hat{m}_1),\mathbf{y}_{1\mathcal{B}}^n\big)\in\mathcal{T}_\epsilon^{(n)}$.
		
		\textit{State estimation:} The sensing Rx estimates the state symbol by symbol using the decoded common-layer codeword, i.e.,
		\begin{equation*}
			\hat{S}_{k,t}=\hat{s}_k^*\big(u_{k,t}(\hat{m}_0),Y_{2k,t}\big),\quad k\in[1:K],\ t\in[1:n].
		\end{equation*}
		
		\textit{Analysis of the probability of error:} By the standard analysis of superposition coding~\cite[Sec.~5.3]{ElGamalKim2011} and the product structure of the channel, the decoding error probabilities at both receivers tend to zero as $n\to\infty$ if
		\begin{align}
			R_u&<\sum_{k=1}^{K}I(U_k;Y_{2k})-\delta(\epsilon),\label{eq:par-ach-sen}\\
			R_u&<\sum_{k=1}^{K}I(U_k;Y_{1k})-\delta(\epsilon),\label{eq:par-ach-com0}\\
			R_v&<\sum_{k\in\mathcal{B}}I(X_k;Y_{1k}|U_k)-\delta(\epsilon),\label{eq:par-ach-com1}
		\end{align}
		where $\delta(\epsilon)\to0$ as $\epsilon\to0$. A constraint among~\eqref{eq:par-ach-sen}--\eqref{eq:par-ach-com1} is imposed only when the rate of the corresponding message layer is positive; if a layer has zero rate, its message set is a singleton and no decoding is required for that layer. Since $U_k\leftrightarrow X_k\leftrightarrow Y_{1k}$, we have $I(U_k;Y_{1k})+I(X_k;Y_{1k}|U_k)=I(X_k;Y_{1k})$. Therefore, for any $R$ satisfying
		\begin{align}
			R<\min\Bigg\{
			&\sum_{k=1}^{K}I(X_k;Y_{1k}),
			\sum_{k\in\mathcal{A}}I(X_k;Y_{2k})\notag\\
			&+\sum_{k\in\mathcal{B}}\big[I(X_k;Y_{1k}|U_k)+I(U_k;Y_{2k})\big]
			\Bigg\}-2\delta(\epsilon),
			\label{eq:par-ach-rate}
		\end{align}
		one can choose $R_u,R_v\ge0$ such that~\eqref{eq:par-ach-sen}--\eqref{eq:par-ach-com1} hold.
		
		\textit{Analysis of the distortion:} Let $\mathcal{E}_0:=\{\hat{M}_0\neq M_0\}$. On $\mathcal{E}_0^c$, the estimate at the sensing Rx coincides with the genie-aided estimate $\hat{s}_k^*(U_{k,t},Y_{2k,t})$ that knows the common message. Averaged over the random codebook, the codebook generation and the product structure of the channel imply that $(U_{k,t},Y_{2k,t},S_{k,t})\sim p_{U_kY_{2k}S_k}$; hence, the average distortion of the genie-aided estimate is
		$\bar{D}:=\sum_{k=1}^{K}d_k^*(S_k|U_k,Y_{2k})\le D$.
		Since $d_{\max}<\infty$, we have
		$\mathbb{E}_{\mathcal{C}}\big[D^{(n)}(\mathcal{C})\big]
			\le \bar{D}+d_{\max}\,\mathbb{E}_{\mathcal{C}}\big[\mathbb{P}(\mathcal{E}_0|\mathcal{C})\big]$,
		where $\mathcal{C}$ denotes the random codebook, and the second term on the right-hand side tends to zero as $n\to\infty$ by~\eqref{eq:par-ach-sen}.
		
		Finally, let $P_e(\mathcal{C})$ denote the communication decoding error probability of codebook $\mathcal{C}$, and choose $\lambda_n\to\infty$ such that $\lambda_n\mathbb{E}_{\mathcal{C}}[P_e(\mathcal{C})]\to0$. Then there exists a codebook $\mathcal{C}_n$ such that $D^{(n)}(\mathcal{C}_n)+\lambda_nP_e(\mathcal{C}_n)\le\mathbb{E}_{\mathcal{C}}\big[D^{(n)}(\mathcal{C})+\lambda_nP_e(\mathcal{C})\big]$, which implies that $P_e(\mathcal{C}_n)\to0$ and $\limsup_{n\to\infty}D^{(n)}(\mathcal{C}_n)\le\bar{D}\le D$. Letting $\epsilon\to0$ completes the proof of the achievability part of Theorem~\ref{thm:parallel}.
		
		\subsection{Converse}
		
		By Fano's inequality~\cite{ElGamalKim2011}, we have
		\begin{align}\label{eq:par-conv-fano}
			nR-n\epsilon_n
			\leq I(M;\mathbf{Y}_1^n)\overset{(a)}{\leq} I(\mathbf{X}^n;\mathbf{Y}_1^n)\overset{(b)}{\leq}\sum_{t=1}^{n}I(\mathbf{X}_t;\mathbf{Y}_{1,t}),
		\end{align}
		where $(a)$ follows from the data processing inequality and $(b)$ follows from the memorylessness of the channel~\eqref{eq:model-memoryless}.
		
		On the other hand, let $\mathbf O_t:=(\mathbf Y_{1,t},\mathbf Y_{2,t})$ and write $\mathbf O^n=(\mathbf O_1,\ldots,\mathbf O_n)$. Define $\mathbf W_t:=(\mathbf O^{t-1},\mathbf Y_{2,t+1}^n)$. By Fano's inequality, we have
		\begin{align}
			nR-n\epsilon_n&\le I(\mathbf{X}^n;\mathbf{O}^n)\notag\\
			&=\sum_{t=1}^{n}I(\mathbf{X}^n;\mathbf{O}_t|\mathbf{O}^{t-1})\notag\\
			&\overset{(a)}{=}\sum_{t=1}^{n}I(\mathbf{X}_t,\mathbf{Y}_{2,t+1}^n;\mathbf{O}_t|\mathbf{O}^{t-1})\notag\\
			&=\sum_{t=1}^{n}\big[I(\mathbf{X}_t;\mathbf{O}_t|\mathbf{W}_t)+I(\mathbf{Y}_{2,t+1}^n;\mathbf{O}_t|\mathbf{O}^{t-1})\big]\notag\\
			&\overset{(b)}{=}\sum_{t=1}^{n}\big[I(\mathbf{X}_t;\mathbf{O}_t|\mathbf{W}_t)+I(\mathbf{O}^{t-1};\mathbf{Y}_{2,t}|\mathbf{Y}_{2,t+1}^n)\big]\notag\\
			&\overset{(c)}{\le}\sum_{t=1}^{n}\big[I(\mathbf{X}_t;\mathbf{O}_t|\mathbf{W}_t)+I(\mathbf{W}_t;\mathbf{Y}_{2,t})\big],\label{eq:par-conv-temporal}
		\end{align}
		where $(a)$ follows from the Markov chain $(\mathbf{X}^{n\setminus t},\mathbf{O}^{t-1},\mathbf{Y}_{2,t+1}^n)\leftrightarrow\mathbf{X}_t\leftrightarrow(\mathbf{O}_t,\mathbf{S}_t)$, $(b)$ follows from the Csisz\'ar sum identity, and $(c)$ follows from $I(\mathbf{O}^{t-1};\mathbf{Y}_{2,t}|\mathbf{Y}_{2,t+1}^n)=I(\mathbf{W}_t;\mathbf{Y}_{2,t})-I(\mathbf{Y}_{2,t+1}^n;\mathbf{Y}_{2,t})\le I(\mathbf{W}_t;\mathbf{Y}_{2,t})$.
		
		We introduce a time-sharing random variable $Q\sim\mathrm{Unif}[1:n]$, independent of $(\mathbf X^n,\mathbf S^n,\mathbf Y_1^n,\mathbf Y_2^n)$, and define $\mathbf W:=(Q,\mathbf W_Q)$, $\mathbf X:=\mathbf X_Q$, $\mathbf Y_i:=\mathbf Y_{i,Q}$ ($i=1,2$), $\mathbf S:=\mathbf S_Q$, and $\mathbf O:=\mathbf O_Q$.
		
		From~\eqref{eq:par-conv-fano}, we obtain
		\begin{align}\label{eq:par-conv-rate1}
			R-\epsilon_n
			&\leq I(\mathbf X;\mathbf Y_1|Q)\nonumber\\
			&\leq I(\mathbf X;\mathbf Y_1)\nonumber\\
			&\overset{(a)}{\leq}\sum_{k=1}^{K}I(X_k;Y_{1k}),
		\end{align}
		where $(a)$ follows from the product structure~\eqref{eq:par-channel}.
		
		From~\eqref{eq:par-conv-temporal}, we obtain
		\begin{align}\label{eq:par-conv-single}
			R-\epsilon_n
			&\leq I(\mathbf X;\mathbf O|\mathbf W_Q,Q)+I(\mathbf W_Q;\mathbf Y_2|Q)\nonumber\\
			&=I(\mathbf X;\mathbf O|\mathbf W)+I(\mathbf W_Q;\mathbf Y_2|Q)\nonumber\\
			&\overset{(a)}{\leq}I(\mathbf X;\mathbf O|\mathbf W)+I(\mathbf W;\mathbf Y_2),
		\end{align}
		where $(a)$ follows from $I(\mathbf W_Q;\mathbf Y_2|Q)\leq I(Q,\mathbf W_Q;\mathbf Y_2)=I(\mathbf W;\mathbf Y_2)$.
		
		Next, we simplify~\eqref{eq:par-conv-single} by exploiting the degradedness. For $i=1,2$, let $\mathbf{Y}_{i\mathcal{A}}:=(Y_{ik})_{k\in\mathcal{A}}$ and $\mathbf{Y}_{i\mathcal{B}}:=(Y_{ik})_{k\in\mathcal{B}}$, and define $\mathbf{X}_{\mathcal{A}}$ and $\mathbf{X}_{\mathcal{B}}$ similarly. Let $\mathbf{O}^+:=(\mathbf{Y}_{2\mathcal{A}},\mathbf{Y}_{1\mathcal{B}})$ and $\mathbf{O}^-:=(\mathbf{Y}_{1\mathcal{A}},\mathbf{Y}_{2\mathcal{B}})$ denote the stronger and weaker outputs of the component channels, respectively, and let $\bar{\mathbf{W}}:=(\mathbf{W},\mathbf{Y}_{2\mathcal{A}})$. Then
		\begin{align}\label{eq:par-conv-peel}
			&I(\mathbf{X};\mathbf{O}|\mathbf{W})+I(\mathbf{W};\mathbf{Y}_2)\notag\\
			&=I(\mathbf{X};\mathbf{O}^+|\mathbf{W})+I(\mathbf{X};\mathbf{O}^-|\mathbf{W},\mathbf{O}^+)+I(\mathbf{W};\mathbf{Y}_2)\notag\\
			&\overset{(a)}{=}I(\mathbf{X};\mathbf{Y}_{2\mathcal{A}},\mathbf{Y}_{1\mathcal{B}}|\mathbf{W})+I(\mathbf{W};\mathbf{Y}_{2\mathcal{A}},\mathbf{Y}_{2\mathcal{B}})\notag\\
			&=I(\mathbf{X},\mathbf{W};\mathbf{Y}_{2\mathcal{A}})+I(\mathbf{X};\mathbf{Y}_{1\mathcal{B}}|\bar{\mathbf{W}})+I(\mathbf{W};\mathbf{Y}_{2\mathcal{B}}|\mathbf{Y}_{2\mathcal{A}})\notag\\
			&\overset{(b)}{=}I(\mathbf{X}_{\mathcal{A}};\mathbf{Y}_{2\mathcal{A}})+I(\mathbf{X}_{\mathcal{B}};\mathbf{Y}_{1\mathcal{B}}|\bar{\mathbf{W}})+I(\mathbf{W};\mathbf{Y}_{2\mathcal{B}}|\mathbf{Y}_{2\mathcal{A}})\notag\\
			&\overset{(c)}{\le}I(\mathbf{X}_{\mathcal{A}};\mathbf{Y}_{2\mathcal{A}})
			+\underbrace{I(\mathbf{X}_{\mathcal{B}};\mathbf{Y}_{1\mathcal{B}}|\bar{\mathbf{W}})+I(\bar{\mathbf{W}};\mathbf{Y}_{2\mathcal{B}})}_{\Psi_{\mathcal{B}}}\notag\\
			&{\le}\sum_{k\in\mathcal{A}}I(X_k;Y_{2k})+\Psi_{\mathcal{B}},
		\end{align}
		where $(a)$ holds because, by the product structure~\eqref{eq:par-channel} of the channel and the degradedness condition~\eqref{eq:par-degraded},
		\begin{equation*}
			p_{\mathbf{O}^-|\mathbf{X}\mathbf{O}^+}
			=\prod_{k\in\mathcal{A}}p_{Y_{1k}|Y_{2k}}\prod_{k\in\mathcal{B}}p_{Y_{2k}|Y_{1k}}
			=p_{\mathbf{O}^-|\mathbf{O}^+},
		\end{equation*}
		which, together with $\mathbf{W}\leftrightarrow\mathbf{X}\leftrightarrow\mathbf{O}$, yields the Markov chain $(\mathbf{W},\mathbf{X})\leftrightarrow\mathbf{O}^+\leftrightarrow\mathbf{O}^-$, and hence $I(\mathbf{X};\mathbf{O}^-|\mathbf{W},\mathbf{O}^+)=0$; $(b)$ follows from the product structure of the channel and the Markov chains $(\mathbf{W},\mathbf{X}_{\mathcal{B}})\leftrightarrow\mathbf{X}_{\mathcal{A}}\leftrightarrow\mathbf{Y}_{2\mathcal{A}}$ and $(\bar{\mathbf{W}},\mathbf{X}_{\mathcal{A}})\leftrightarrow\mathbf{X}_{\mathcal{B}}\leftrightarrow\mathbf{Y}_{1\mathcal{B}}$, which are implied by $\mathbf{W}\leftrightarrow\mathbf{X}\leftrightarrow\mathbf{O}$; and $(c)$ follows since $I(\mathbf{W};\mathbf{Y}_{2\mathcal{B}}|\mathbf{Y}_{2\mathcal{A}})=I(\bar{\mathbf{W}};\mathbf{Y}_{2\mathcal{B}})-I(\mathbf{Y}_{2\mathcal{A}};\mathbf{Y}_{2\mathcal{B}})\le I(\bar{\mathbf{W}};\mathbf{Y}_{2\mathcal{B}})$.
		
		Next, we bound $\Psi_{\mathcal{B}}$ component by component. For $k\in\mathcal{B}$, let $\widetilde{\mathbf{O}}_k:=(Y_{1k},Y_{2k})$, $\widetilde{\mathbf{O}}_{\mathcal{B}}:=(\widetilde{\mathbf{O}}_k)_{k\in\mathcal{B}}$, $\widetilde{\mathbf{O}}_{\mathcal{B}}^{<k}:=(\widetilde{\mathbf{O}}_j)_{j\in\mathcal{B},\,j<k}$, $\mathbf{Y}_{2\mathcal{B}}^{>k}:=(Y_{2j})_{j\in\mathcal{B},\,j>k}$, and $U_k:=(\bar{\mathbf{W}},\widetilde{\mathbf{O}}_{\mathcal{B}}^{<k},\mathbf{Y}_{2\mathcal{B}}^{>k})$. Then
		\begin{align}\label{eq:par-conv-spatial}
			\Psi_{\mathcal{B}}
			&=I(\mathbf{X}_{\mathcal{B}};\mathbf{Y}_{1\mathcal{B}}|\bar{\mathbf{W}})+I(\bar{\mathbf{W}};\mathbf{Y}_{2\mathcal{B}})\notag\\
			&\overset{(a)}{\le}I(\mathbf{X}_{\mathcal{B}};\widetilde{\mathbf{O}}_{\mathcal{B}}|\bar{\mathbf{W}})+I(\bar{\mathbf{W}};\mathbf{Y}_{2\mathcal{B}})\notag\\
			&\overset{(b)}{=}\sum_{k\in\mathcal{B}}\big[I(\mathbf{X}_{\mathcal{B}};\widetilde{\mathbf{O}}_k|\bar{\mathbf{W}},\widetilde{\mathbf{O}}_{\mathcal{B}}^{<k})+I(\bar{\mathbf{W}};Y_{2k}|\mathbf{Y}_{2\mathcal{B}}^{>k})\big]\notag\\
			&\le\sum_{k\in\mathcal{B}}\big[I(\mathbf{X}_{\mathcal{B}},\mathbf{Y}_{2\mathcal{B}}^{>k};\widetilde{\mathbf{O}}_k|\bar{\mathbf{W}},\widetilde{\mathbf{O}}_{\mathcal{B}}^{<k})+I(\bar{\mathbf{W}};Y_{2k}|\mathbf{Y}_{2\mathcal{B}}^{>k})\big]\notag\\
			&=\sum_{k\in\mathcal{B}}\big[I(\mathbf{X}_{\mathcal{B}};\widetilde{\mathbf{O}}_k|\bar{\mathbf{W}},\widetilde{\mathbf{O}}_{\mathcal{B}}^{<k},\mathbf{Y}_{2\mathcal{B}}^{>k})\notag\\
			&\qquad+I(\mathbf{Y}_{2\mathcal{B}}^{>k};\widetilde{\mathbf{O}}_k|\bar{\mathbf{W}},\widetilde{\mathbf{O}}_{\mathcal{B}}^{<k})+I(\bar{\mathbf{W}};Y_{2k}|\mathbf{Y}_{2\mathcal{B}}^{>k})\big]\notag\\
			&\overset{(c)}{=}\sum_{k\in\mathcal{B}}\big[I(\mathbf{X}_{\mathcal{B}};\widetilde{\mathbf{O}}_k|\bar{\mathbf{W}},\widetilde{\mathbf{O}}_{\mathcal{B}}^{<k},\mathbf{Y}_{2\mathcal{B}}^{>k})\notag\\
			&\qquad+I(\widetilde{\mathbf{O}}_{\mathcal{B}}^{<k};Y_{2k}|\bar{\mathbf{W}},\mathbf{Y}_{2\mathcal{B}}^{>k})+I(\bar{\mathbf{W}};Y_{2k}|\mathbf{Y}_{2\mathcal{B}}^{>k})\big]\notag\\
			&\le\sum_{k\in\mathcal{B}}\big[I(\mathbf{X}_{\mathcal{B}};\widetilde{\mathbf{O}}_k|U_k)+I(U_k;Y_{2k})\big]\notag\\
			&\overset{(d)}{=}\sum_{k\in\mathcal{B}}\big[I(X_k;Y_{1k}|U_k)+I(U_k;Y_{2k})\big],
		\end{align}
		where $(a)$ follows by adding $\mathbf{Y}_{2\mathcal{B}}$ to the observation; $(b)$ follows from the chain rule, with the first term expanded in increasing order of $k$ and the second term in decreasing order of $k$; $(c)$ follows from the Csisz\'ar sum identity; and $(d)$ holds because, by the product structure~\eqref{eq:par-channel} and $\mathbf{W}\leftrightarrow\mathbf{X}\leftrightarrow(\mathbf{Y}_1,\mathbf{Y}_2,\mathbf{S})$, we have $(U_k,\mathbf{X}^{K\setminus k})\leftrightarrow X_k\leftrightarrow(Y_{1k},Y_{2k},S_k)$, which, combined with $X_k\leftrightarrow Y_{1k}\leftrightarrow Y_{2k}$, yields $I(\mathbf{X}_{\mathcal{B}};\widetilde{\mathbf{O}}_k|U_k)=I(X_k;\widetilde{\mathbf{O}}_k|U_k)=I(X_k;Y_{1k}|U_k)$.
		
		Combining~\eqref{eq:par-conv-single}, \eqref{eq:par-conv-peel}, and~\eqref{eq:par-conv-spatial}, we obtain
		\begin{equation}\label{eq:par-conv-rate2}
			R-\epsilon_n\leq\sum_{k\in\mathcal{A}}I(X_k;Y_{2k})+\sum_{k\in\mathcal{B}}\big[I(X_k;Y_{1k}|U_k)+I(U_k;Y_{2k})\big].
		\end{equation}
		
		We now consider the distortion constraint. For $k\in\mathcal{B}$ and $t\in[1:n]$, let $U_{k,t}:=(\mathbf{W}_t,\mathbf{Y}_{2\mathcal{A},t},\widetilde{\mathbf{O}}_{\mathcal{B},t}^{<k},\mathbf{Y}_{2\mathcal{B},t}^{>k})$, where $\widetilde{\mathbf{O}}_{\mathcal{B},t}^{<k}$ and $\mathbf{Y}_{2\mathcal{B},t}^{>k}$ denote the components of $\widetilde{\mathbf{O}}_{\mathcal{B}}^{<k}$ and $\mathbf{Y}_{2\mathcal{B}}^{>k}$ in the $t$th channel use, respectively. Accordingly, the auxiliary random variable in~\eqref{eq:par-conv-spatial} is $U_k=(Q,U_{k,Q})$. Since the estimate $\hat{\mathbf{S}}_t$ depends only on $\mathbf{Y}_2^n$ and conditioning does not increase the minimum conditional expected distortion, we have, for $k\in\mathcal{B}$,
		\begin{align}\label{eq:conv-dist-B}
			\mathbb{E}\big[d_k(S_{k,t},\hat{S}_{k,t})\big]
			&\ge d_k^*(S_{k,t}|\mathbf{Y}_2^n)\notag\\
			&\ge d_k^*\big(S_{k,t}\,\big|\,\mathbf{Y}_1^{t-1},\mathbf{Y}_2^n,\mathbf{Y}_{1\mathcal{B},t}^{<k}\big)\notag\\
			&=d_k^*(S_{k,t}|U_{k,t},Y_{2k,t}),
		\end{align}
		where $\mathbf{Y}_{1\mathcal{B},t}^{<k}:=(Y_{1j,t})_{j\in\mathcal{B},\,j<k}$, and the last equality holds because $(\mathbf{Y}_1^{t-1},\mathbf{Y}_2^n,\mathbf{Y}_{1\mathcal{B},t}^{<k})$ and $(U_{k,t},Y_{2k,t})$ consist of the same random variables. Similarly, for $k\in\mathcal{A}$,
		\begin{align}\label{eq:conv-dist-A}
			\mathbb{E}\big[d_k(S_{k,t},\hat{S}_{k,t})\big]
			&\ge d_k^*(S_{k,t}|\mathbf{X}^n,\mathbf{Y}_2^n)\notag\\
			&=d_k^*(S_{k,t}|X_{k,t},Y_{2k,t}),
		\end{align}
		where the equality follows from the Markov chain $S_{k,t}\leftrightarrow(X_{k,t},Y_{2k,t})\leftrightarrow(\mathbf{X}^n,\mathbf{Y}_2^n)$. Therefore,
		\begin{align}\label{eq:conv-dist}
			D+\epsilon
			&\ge\frac{1}{n}\sum_{t=1}^{n}\sum_{k=1}^{K}\mathbb{E}\big[d_k(S_{k,t},\hat{S}_{k,t})\big]\notag\\
			&\ge\frac{1}{n}\sum_{t=1}^{n}\bigg[\sum_{k\in\mathcal{A}}d_k^*(S_{k,t}|X_{k,t},Y_{2k,t})\notag\\
			&\qquad+\sum_{k\in\mathcal{B}}d_k^*(S_{k,t}|U_{k,t},Y_{2k,t})\bigg]\notag\\
			&=\sum_{k\in\mathcal{A}}d_k^*(S_k|X_k,Y_{2k})+\sum_{k\in\mathcal{B}}d_k^*(S_k|U_k,Y_{2k}),
		\end{align}
		where the last equality follows from $U_k=(Q,U_{k,Q})$ and the Markov chain $Q\leftrightarrow(X_k,Y_{2k})\leftrightarrow S_k$, with $S_k:=S_{k,Q}$, $X_k:=X_{k,Q}$, and $Y_{2k}:=Y_{2k,Q}$.
		
		Combining~\eqref{eq:par-conv-rate1}, \eqref{eq:par-conv-rate2}, and~\eqref{eq:conv-dist}, we conclude that $R-\epsilon_n$ is upper-bounded by the maximum of the right-hand side of Theorem~\ref{thm:parallel} under the distortion constraint $D+\epsilon$. Note that the resulting inputs on different component channels may be correlated. However, since the rate expressions and the distortion constraint depend only on the marginal distributions associated with the individual component channels, the joint distribution can be replaced by the product of these marginals without changing the corresponding rates and distortion. Finally, letting $n\to\infty$ and $\epsilon\to0$ completes the proof of the converse.
		\begin{remark}\label{rem:par-seq}
			Theorem~\ref{thm:parallel} is established for the symbol-wise distortion~\eqref{eq:model-distortion}. If the distortion is replaced by the sequence-wise logarithmic loss in~\cite{Chen2026TCOM}, whose minimum is $\frac{1}{n}H(\mathbf{S}^n|\mathbf{Y}_2^n)$, the chain rule introduces past states into the conditioning, so that $\mathbf{S}_{\mathcal{B},t}$ and $S_k$ have to be included in $\mathbf{O}_t$ and $\widetilde{\mathbf{O}}_k$, respectively. Removing these states in~\eqref{eq:par-conv-peel} and~\eqref{eq:par-conv-spatial} then requires the stronger Markov chain $X_k\leftrightarrow Y_{1k}\leftrightarrow(Y_{2k},S_k)$ for $k\in\mathcal{B}$, whereas no additional condition is needed for $k\in\mathcal{A}$. Under this condition, the resulting $C(D)$ coincides with that under the symbol-wise logarithmic loss, i.e., Theorem~\ref{thm:parallel} with $d_k^*(S_k|\cdot)$ replaced by $H(S_k|\cdot)$, in agreement with~\cite[Corollary~2 and Remark~5]{Chen2026TCOM}.
		\end{remark}

		\subsection{A Special Gaussian Case}\label{subsec:par-gauss}
		
		Consider a parallel bistatic ISAC channel with additive Gaussian states, whose $k$th component channel is given by
		\begin{align}
			Y_{1k}&=X_k+Z_{1k},\\
			Y_{2k}&=X_k+Z_{2k}+S_k,
		\end{align}
		where $Z_{1k}\sim\mathcal{N}(0,N_{1k})$, $Z_{2k}\sim\mathcal{N}(0,N_{2k})$, and $S_k\sim\mathcal{N}(0,N_{Sk})$ are mutually independent and independent of the channel input; moreover, the noises and states of different component channels are mutually independent.
		
		The sensing distortion is measured by the squared error $d_k(s_k,\hat{s}_k)=(s_k-\hat{s}_k)^2$, in which case the minimum conditional expected distortion in~\eqref{eq:model-dstar} becomes
		\begin{equation}\label{eq:par-g-var}
			d_k^*(S_k|{\Theta})=\mathrm{var}(S_k|{\Theta}).
		\end{equation}
		In addition, Definition~\ref{def:achievable} is modified to further impose the input power constraint
		\begin{equation}\label{eq:par-g-power}
			\frac{1}{n}\sum_{t=1}^{n}\sum_{k=1}^{K}\mathbb{E}\big[X_{k,t}^2\big]\le P+\epsilon,
		\end{equation}
		and the supremum of the corresponding achievable rates, denoted by $C(D,P)$, is referred to as the capacity-distortion-power function.
		
		By Remark~\ref{rem:marginal}, stochastic degradedness and physical degradedness lead to the same $C(D)$. Therefore, for the Gaussian channel under consideration, the direction of degradation can be determined simply by comparing the effective noise variances at the communication Rx and the sensing Rx. Specifically,
		\begin{align*}
			\mathcal A&=\{k:N_{1k}\geq N_{2k}+N_{Sk}\},\\
			\mathcal B&=\{k:N_{1k}<N_{2k}+N_{Sk}\}.
		\end{align*}
		
		By extending Theorem~\ref{thm:parallel} to the Gaussian case,\footnote{The achievability result for the discrete case can be extended to the Gaussian case via the standard discretization argument, and the converse can be established using the weak convergence argument in~\cite{Geng2014IT}.} we obtain
		\begin{subequations}\label{eq:par-g-opt}
			\begin{align}
				C(D,P)=\sup\;&\min\{R_1,R_2\}\label{eq:par-g-opt-a}\\
				\text{s.t.}\;&\sum_{k\in\mathcal{A}}\mathrm{var}(S_k|X_k,Y_{2k})\notag\\
				&\quad+\sum_{k\in\mathcal{B}}\mathrm{var}(S_k|U_k,Y_{2k})\le D,\label{eq:par-g-opt-b}\\
				&\sum_{k=1}^{K}\mathbb{E}[X_k^2]\le P.\label{eq:par-g-opt-c}
			\end{align}
		\end{subequations}

		By the optimality result for scalar Gaussian channels in~\cite[Theorem~5]{Chen2026TCOM}, the optimization over input distributions in~\eqref{eq:par-g-opt} can be restricted to Gaussian superposition inputs. For $k\in\mathcal{A}$, the distortion is $\mathrm{var}(S_k|X_k,Y_{2k})=\frac{N_{2k}N_{Sk}}{N_{2k}+N_{Sk}}$. For $k\in\mathcal{B}$, given the power $P_k$ and the distortion $d_k$ of the component channel, the corresponding optimal rate is determined by the above scalar Gaussian result. Therefore, the computation of $C(D,P)$ reduces to the following joint power and distortion allocation problem across component channels:
		\begin{subequations}\label{eq:par-g-convex}
			\begin{align}
				C(D,P)=\sup_{\{P_k,d_k\}}\quad
				&\min\{{R}_1,{R}_2\}\label{eq:par-g-convex-a}\\
				\text{s.t.}\quad
				&\sum_{k\in\mathcal{A}}\frac{N_{2k}N_{Sk}}{N_{2k}+N_{Sk}}+\sum_{k\in\mathcal{B}}d_k\leq D,\label{eq:par-g-convex-b}\\
				&\frac{N_{2k}N_{Sk}}{N_{2k}+N_{Sk}}\leq d_k\leq\frac{(P_k+N_{2k})N_{Sk}}{P_k+N_{2k}+N_{Sk}},\label{eq:par-g-convex-c}\\
				&\sum_{k=1}^{K}P_k\leq P,\qquad P_k\geq0,\label{eq:par-g-convex-d}
			\end{align}
		\end{subequations}
		where
		\begin{align*}
			{R}_1&=\sum_{k=1}^{K}\frac{1}{2}\log\Big(1+\frac{P_k}{N_{1k}}\Big),\\
			{R}_2&=\sum_{k\in\mathcal{A}}\frac{1}{2}\log\Big(1+\frac{P_k}{N_{2k}+N_{Sk}}\Big)\\
			&\quad+\sum_{k\in\mathcal{B}}\frac{1}{2}\log\frac{P_k+N_{2k}+N_{Sk}}{N_{Sk}}\\
			&\quad+\sum_{k\in\mathcal{B}}\frac{1}{2}\log\frac{(N_{1k}-N_{2k})N_{Sk}+(N_{2k}+N_{Sk}-N_{1k})d_k}{N_{1k}N_{Sk}}.
		\end{align*}
		
		Since $R_1$ and $R_2$ are jointly concave in $\{P_k,d_k\}$ and the feasible set of~\eqref{eq:par-g-convex} is convex, \eqref{eq:par-g-convex} can be solved by Lagrangian dual decomposition. Introducing an auxiliary variable $\tau$, \eqref{eq:par-g-convex} is equivalent to maximizing $\tau$ subject to $\tau\le R_1$, $\tau\le R_2$, and~\eqref{eq:par-g-convex-b}--\eqref{eq:par-g-convex-d}. The stationarity condition with respect to $\tau$ implies that the Lagrange multipliers associated with $\tau\le R_1$ and $\tau\le R_2$ are $\theta$ and $1-\theta$ for some $\theta\in[0,1]$. With $\lambda,\mu\ge0$ associated with the power and distortion constraints, the Lagrangian is
		\begin{align}\label{eq:par-wf-lag}
			\mathcal{L}
			&=\theta R_1+(1-\theta)R_2-\lambda\Big(\sum_{k=1}^{K}P_k-P\Big)
			-\mu\Big(\sum_{k=1}^{K}d_k-D\Big),
		\end{align}
		where $d_k$ with $k\in\mathcal{A}$ is also regarded as a variable satisfying $d_k\ge\frac{N_{2k}N_{Sk}}{N_{2k}+N_{Sk}}$. The dual function and the dual problem are
		\begin{equation}\label{eq:par-wf-dual}
			g(\theta,\lambda,\mu)=\max_{\{P_k,d_k\}}\mathcal{L},\quad
			\min_{\theta\in[0,1],\,\lambda\ge0,\,\mu\ge0}g(\theta,\lambda,\mu),
		\end{equation}
		where the maximization is over $P_k\ge0$ and $d_k$ satisfying~\eqref{eq:par-g-convex-c}. Strong duality holds by Slater's condition, and the maximization in~\eqref{eq:par-wf-dual} decomposes across the component channels. Given the per-channel solution $\{P_k^*,d_k^*\}$ derived below, the dual variables $(\theta,\lambda,\mu)$ are obtained by solving the dual problem in~\eqref{eq:par-wf-dual}, e.g., by bisection. For brevity, let $a_k:=N_{2k}+N_{Sk}-N_{1k}$, which is positive for $k\in\mathcal{B}$.
		
		\subsubsection{Optimal Power Allocation}
		The partial derivative of $\mathcal{L}$ with respect to $P_k$ is
		\begin{equation}\label{eq:par-wf-dP}
			\frac{\partial\mathcal{L}}{\partial P_k}
			=\frac{1}{2\ln2}\bigg(\frac{\theta}{P_k+N_{1k}}+\frac{1-\theta}{P_k+N_{2k}+N_{Sk}}\bigg)-\lambda.
		\end{equation}
		Setting~\eqref{eq:par-wf-dP} to zero and combining with $P_k\ge0$, we obtain the water-filling solution
		\begin{equation}\label{eq:par-wf-power}
			P_k^*=\bigg[\frac{\nu_P-N_{1k}-N_{2k}-N_{Sk}+\sqrt{\Delta_k}}{2}\bigg]^+,
		\end{equation}
		where $\nu_P:=\frac{1}{2\lambda\ln2}$ is the power water level and $\Delta_k:=\big(\nu_P-(1-2\theta)a_k\big)^2+4\theta(1-\theta)a_k^2$. For $\theta=1$ and $\theta=0$, \eqref{eq:par-wf-power} reduces to the classical water-filling over $N_{1k}$ and over $N_{2k}+N_{Sk}$, respectively.
		
		\subsubsection{Optimal Distortion Allocation}
		The partial derivative of $\mathcal{L}$ with respect to $d_k$ is
		\begin{equation}\label{eq:par-wf-dd}
			\frac{\partial\mathcal{L}}{\partial d_k}=
			\begin{cases}
				-\mu, & k\in\mathcal{A},\\[1mm]
				\dfrac{(1-\theta)a_k}{2\ln2\,\big[(N_{1k}-N_{2k})N_{Sk}+a_kd_k\big]}-\mu, & k\in\mathcal{B}.
			\end{cases}
		\end{equation}
		Hence, $d_k^*$ attains its lower bound for $k\in\mathcal{A}$, and, for $k\in\mathcal{B}$, setting~\eqref{eq:par-wf-dd} to zero yields
		\begin{equation}\label{eq:par-wf-dist}
			d_k^*=
			\begin{cases}
				\dfrac{N_{2k}N_{Sk}}{N_{2k}+N_{Sk}}, & k\in\mathcal{A},\\[2mm]
				\Pi_k\bigg((1-\theta)\nu_D-\dfrac{(N_{1k}-N_{2k})N_{Sk}}{a_k}\bigg), & k\in\mathcal{B},
			\end{cases}
		\end{equation}
		where $\nu_D:=\frac{1}{2\mu\ln2}$ is the distortion water level and $\Pi_k(\cdot)$ denotes the projection onto the interval in~\eqref{eq:par-g-convex-c}. If the upper bound in~\eqref{eq:par-g-convex-c} is active, component channel $k\in\mathcal{B}$ carries only private information, its weighted rate reduces to $\frac{1}{2}\log(1+\frac{P_k}{N_{1k}})$ regardless of $\theta$, and $P_k^*$ is instead obtained from the cubic equation
		\begin{equation}\label{eq:par-wf-coupled}
			\frac{1}{P_k+N_{1k}}=\frac{1}{\nu_P}+\frac{1}{\nu_D}\cdot\frac{N_{Sk}^2}{(P_k+N_{2k}+N_{Sk})^2}.
		\end{equation}
		
		Fig.~\ref{fig:joint-vs-separate} compares the capacity-distortion performance of joint coding across component channels (Theorem~\ref{thm:parallel}) with that of separate coding (Remark~\ref{rem:par-gain}) for $K=2$, where the two component channels belong to $\mathcal B$ and $\mathcal A$, respectively. It can be seen that joint coding offers a noticeable performance gain, and that the gap between the two schemes gradually decreases as the distortion constraint is relaxed.
		\begin{figure}[t]
			\centering
			\includegraphics[width=0.8\linewidth]{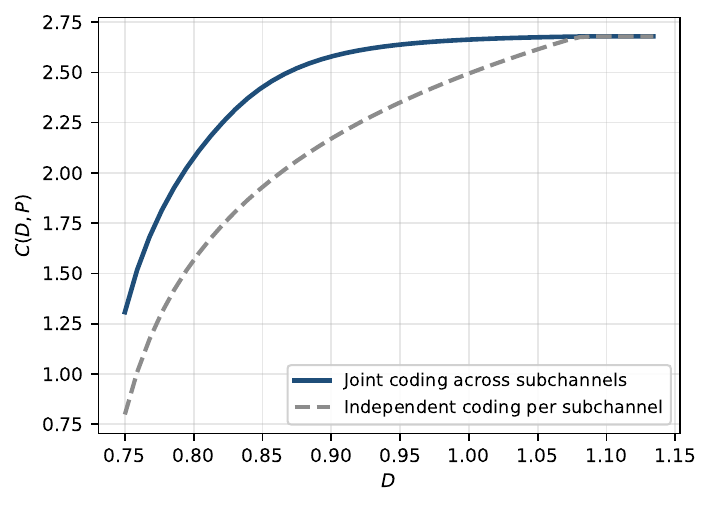}
			\caption{Comparison between joint coding and separate coding for parallel Gaussian ISAC channels under the squared error distortion measure, with $(N_{11},N_{12})=(0.1,5.0)$, $(N_{21},N_{22})=(1.0,0.5)$, $(N_{S1},N_{S2})=(1.0,0.5)$, and $P=4$.}
			\label{fig:joint-vs-separate}
		\end{figure}

		\section{General Vector Channels}\label{sec:vector}
		
		In this section, we consider general vector channels, which are no longer required to have the product structure~\eqref{eq:par-channel} of the parallel channels considered in Section~\ref{sec:parallel}. In systems such as MIMO systems, the components of the input are coupled through the channel, so that degradedness cannot be exploited component by component; moreover, no degradation order in either direction is required between $p_{\mathbf{Y}_1|\mathbf{X}}$ and $p_{\mathbf{Y}_2|\mathbf{X}}$. Section~\ref{subsec:mimo-bounds} derives single-letter lower and upper bounds on $C(D)$ for general vector channels with finite alphabets, and Section~\ref{subsec:mimo-gauss} shows that the lower bound is tight for vector Gaussian channels with additive Gaussian states.
		
		\subsection{Single-Letter Bounds}\label{subsec:mimo-bounds}
		
		\begin{theorem}\label{thm:mimo-ach}
			For general vector ISAC channels, $C(D)\ge\underline{C}(D)$, where
			\begin{subequations}\label{eq:lower-cap}
				\begin{align}
					\underline{C}(D):=\max_{p_{\mathbf{U}\mathbf{X}}}\;
					&\min\big\{I(\mathbf{X};\mathbf{Y}_1),\,I(\mathbf{X};\mathbf{Y}_1|\mathbf{U})+I(\mathbf{U};\mathbf{Y}_2)\big\}\label{eq:lower-cap-a}\\
					\text{s.t.}\;\;
					&d^*(\mathbf{S}|\mathbf{U},\mathbf{Y}_2)\le D.\label{eq:lower-cap-b}
				\end{align}
			\end{subequations}
			Here, the joint distribution factorizes as $p_{\mathbf{U}\mathbf{X}}\,p_{\mathbf{Y}_1\mathbf{Y}_2\mathbf{S}|\mathbf{X}}$, and the auxiliary random variable satisfies $|\mathcal{U}|\le|\mathcal{X}|+1$.
		\end{theorem}
		\begin{IEEEproof}
			Theorem~\ref{thm:mimo-ach} generalizes~\cite[Theorem~3]{Jiao2023bistatic}, where the communication Rx additionally observes the state, i.e., $\mathbf{S}$ is a component of $\mathbf{Y}_1$, and~\cite[Theorem~1]{Chen2026TCOM}, which is established under logarithmic loss. The proof follows the same superposition coding argument as in~\cite{Jiao2023bistatic,Chen2026TCOM} and is omitted.
		\end{IEEEproof}

		\begin{theorem}\label{thm:mimo-conv}
			For general vector ISAC channels, $C(D)\le\overline{C}(D)$, where
			\begin{subequations}\label{eq:upper-cap}
				\begin{align}
					\overline{C}(D):=\max_{p_{\mathbf{U}\mathbf{X}}}\;
					&\min\big\{I(\mathbf{X};\mathbf{Y}_1),\,I(\mathbf{U};\mathbf{Y}_2)+I(\mathbf{X};\mathbf{Y}_1,\mathbf{Y}_2|\mathbf{U})\big\}\label{eq:upper-cap-a}\\
					\text{s.t.}\;\;
					&d^*(\mathbf{S}|\mathbf{U},\mathbf{Y}_2)\le D.\label{eq:upper-cap-b}
				\end{align}
			\end{subequations}
			Here, the joint distribution factorizes as $p_{\mathbf{U}\mathbf{X}}\,p_{\mathbf{Y}_1\mathbf{Y}_2\mathbf{S}|\mathbf{X}}$, and the auxiliary random variable satisfies $|\mathcal{U}|\le|\mathcal{X}|+1$.
		\end{theorem}
		\begin{IEEEproof}
			See Appendix~\ref{app:mimo-conv}.
		\end{IEEEproof}
		\begin{remark}\label{rem:mimo-conv-aux}
			Theorem~\ref{thm:mimo-conv} extends the upper bounds in~\cite[Theorem~2]{Chen2026TCOM} and~\cite[Theorem~5]{Jiao2023bistatic} to the present setting, where the distortion measure is general and the communication Rx need not observe the state. Compared with~\cite{Chen2026TCOM,Jiao2023bistatic}, the bound in Theorem~\ref{thm:mimo-conv} involves a single auxiliary random variable, which is constructed in Appendix~\ref{app:mimo-conv} from $\mathbf{Y}_2^n$ only, whereas the auxiliary random variables in~\cite{Chen2026TCOM,Jiao2023bistatic} also contain past outputs at the communication Rx. This property is essential for Remark~\ref{rem:mimo-coupling} and the channel enhancement argument in Section~\ref{subsec:mimo-gauss}.
		\end{remark}

		\subsection{A Special Gaussian Case}\label{subsec:mimo-gauss}
		
		Consider a vector ISAC channel with additive Gaussian states, whose input-output relationship is given by
		\begin{align}
			\mathbf{Y}_1&=\mathbf{X}+\mathbf{Z}_1,\label{eq:mimo-y1}\\
			\mathbf{Y}_2&=\mathbf{X}+\mathbf{Z}_2+\mathbf{S},\label{eq:mimo-y2}
		\end{align}
		where $\mathbf{X},\mathbf{Y}_1,\mathbf{Y}_2,\mathbf{S}\in\mathbb{R}^K$. The communication noise $\mathbf{Z}_1\sim\mathcal{N}(\mathbf{0},\Sigma_1)$, the sensing noise $\mathbf{Z}_2\sim\mathcal{N}(\mathbf{0},\Sigma_2)$, and the state $\mathbf{S}\sim\mathcal{N}(\mathbf{0},\Sigma_s)$ are mutually independent and independent of the channel input $\mathbf{X}$, and their covariance matrices satisfy $\Sigma_1,\Sigma_2,\Sigma_s\succ 0$. The input is subject to the power constraint $\operatorname{tr}(\Sigma_X)\le P$, where $\Sigma_X:=\mathbb{E}[\mathbf{X}\mathbf{X}^\top]$.
		
		When $\Sigma_1$, $\Sigma_2$, and $\Sigma_s$ are all diagonal, this model reduces to the parallel Gaussian channel discussed in Section~\ref{subsec:par-gauss}. If $\Sigma_1\succeq\Sigma_2+\Sigma_s$, then the communication channel is degraded with respect to the sensing channel; if $\Sigma_1\preceq\Sigma_2+\Sigma_s$, then the opposite holds. In this section, no such positive semidefinite order between the two is required; thus, the communication and sensing channels need not be degraded with respect to each other in either direction.
		
		The sensing distortion is measured by the squared error, in which case the minimum conditional expected distortion in~\eqref{eq:model-dstar} becomes
		\begin{equation}\label{eq:mimo-var}
			d^*(\mathbf{S}|\boldsymbol{\Theta})
			=\operatorname{tr}\operatorname{Cov}(\mathbf{S}|\boldsymbol{\Theta}).
		\end{equation}
		
		By extending Theorems~\ref{thm:mimo-ach} and~\ref{thm:mimo-conv} to the Gaussian case, we obtain
		\begin{align}
			\underline{C}(D,P)\le C(D,P)\le\overline{C}(D,P),
		\end{align}
		where
		\begin{subequations}\label{eq:mimo-ach-bound}
			\begin{align}
				\underline{C}(D,P):=\sup_{p_{\mathbf{U}\mathbf{X}}}\;
				&\min\big\{I(\mathbf{X};\mathbf{Y}_1),\,I(\mathbf{X};\mathbf{Y}_1|\mathbf{U})+I(\mathbf{U};\mathbf{Y}_2)\big\}\label{eq:mimo-ach-bound-a}\\
				\text{s.t.}\;\;
				&\operatorname{tr}\operatorname{Cov}(\mathbf{S}|\mathbf{U},\mathbf{Y}_2)\le D,\label{eq:mimo-ach-bound-b}\\
				&\operatorname{tr}(\Sigma_X)\le P,\label{eq:mimo-ach-bound-c}
			\end{align}
		\end{subequations}
		\begin{subequations}\label{eq:mimo-conv-bound}
			\begin{align}
				\overline{C}(D,P):=\sup_{p_{\mathbf{U}\mathbf{X}}}\;
				&\min\big\{I(\mathbf{X};\mathbf{Y}_1),\,I(\mathbf{X};\mathbf{Y}_1,\mathbf{Y}_2|\mathbf{U})+I(\mathbf{U};\mathbf{Y}_2)\big\}\label{eq:mimo-conv-bound-a}\\
				\text{s.t.}\;\;
				&\operatorname{tr}\operatorname{Cov}(\mathbf{S}|\mathbf{U},\mathbf{Y}_2)\le D,\label{eq:mimo-conv-bound-b}\\
				&\operatorname{tr}(\Sigma_X)\le P,\label{eq:mimo-conv-bound-c}
			\end{align}
		\end{subequations}
		and in both cases the joint distribution factorizes as $p_{\mathbf{U}\mathbf{X}}\,p_{\mathbf{Y}_1\mathbf{Y}_2\mathbf{S}|\mathbf{X}}$.
		\begin{remark}\label{rem:mimo-coupling}In~\eqref{eq:mimo-conv-bound},
			$\mathbf{U}\leftrightarrow\mathbf{X}\leftrightarrow
			(\mathbf{Y}_1,\mathbf{Y}_2,\mathbf{S})$ forms a Markov chain.
			Among the quantities in~\eqref{eq:mimo-conv-bound}, only
			$I(\mathbf{X};\mathbf{Y}_1,\mathbf{Y}_2|\mathbf{U})$ depends
			on the coupling of $\mathbf{Y}_1$ and $\mathbf{Y}_2$ given
			$\mathbf{X}$; the remaining terms and the
			constraints~\eqref{eq:mimo-conv-bound-b}
			and~\eqref{eq:mimo-conv-bound-c} depend only on
			$p_{\mathbf{U}\mathbf{X}}$, $p_{\mathbf{Y}_1|\mathbf{X}}$,
			and $p_{\mathbf{Y}_2\mathbf{S}|\mathbf{X}}$.
			By Remark~\ref{rem:marginal}, $C(D,P)$ is unchanged under
			any recoupling that preserves these two conditional channel
			marginals. Hence, Theorem~\ref{thm:mimo-conv} applies to
			every such coupled channel.
			Moreover, for any fixed code, the joint law of
			$\mathbf{U}=(Q,\mathbf{Y}_2^{Q-1},\mathbf{Y}_{2,Q+1}^n)$
			and $\mathbf{X}=\mathbf{X}_Q$, as constructed in the proof
			of Theorem~\ref{thm:mimo-conv}, depends on the channel only
			through $p_{\mathbf{Y}_2|\mathbf{X}}$ and is therefore
			unchanged under admissible recoupling. Consequently, the
			same code-induced distribution $p_{\mathbf{U}\mathbf{X}}$
			satisfies the converse rate bounds under every admissible
			coupling. We may thus take an infimum over admissible
			couplings inside the supremum over $p_{\mathbf{U}\mathbf{X}}$
			in~\eqref{eq:mimo-conv-bound}, allowing the coupling to
			depend on $p_{\mathbf{U}\mathbf{X}}$.
			This freedom is exploited in the proof of
			$C(D,P)\le R_{\mathrm{G}}(D,P)$, where the coupling is
			constructed using Lemma~\ref{lem:mimo-kkt}.
		\end{remark}
		
		\subsubsection{Proof of the Lower Bound}\label{subsubsec:mimo-ach-proof}
		
		In~\eqref{eq:mimo-ach-bound}, we choose the Gaussian superposition input $\mathbf{X}=\mathbf{U}+\mathbf{V}$, where $\mathbf{U}\sim\mathcal{N}(\mathbf{0},\Sigma_U)$ and $\mathbf{V}\sim\mathcal{N}(\mathbf{0},\Sigma_V)$ are mutually independent and independent of $(\mathbf{Z}_1,\mathbf{Z}_2,\mathbf{S})$, with $\Sigma_U,\Sigma_V\succeq0$, so that $\Sigma_X=\Sigma_U+\Sigma_V$. Let $\Sigma_g:=\Sigma_2+\Sigma_s$ and $\mathbf{G}:=\mathbf{Z}_2+\mathbf{S}\sim\mathcal{N}(\mathbf{0},\Sigma_g)$, so that $\mathbf{Y}_2=\mathbf{X}+\mathbf{G}$. By the differential entropy formula for Gaussian random vectors,
		\begin{align}
			I(\mathbf{X};\mathbf{Y}_1)
			&=h(\mathbf{X}+\mathbf{Z}_1)-h(\mathbf{Z}_1)\notag\\
			&=\frac{1}{2}\log\frac{|\Sigma_U+\Sigma_V+\Sigma_1|}{|\Sigma_1|},\label{eq:mimo-gauss-ach-r1}
		\end{align}
		and
		\begin{align}
			&I(\mathbf{X};\mathbf{Y}_1|\mathbf{U})+I(\mathbf{U};\mathbf{Y}_2)\notag\\
			&=h(\mathbf{V}+\mathbf{Z}_1)-h(\mathbf{Z}_1)+h(\mathbf{X}+\mathbf{G})-h(\mathbf{V}+\mathbf{G})\notag\\
			&=\frac{1}{2}\log\frac{|\Sigma_V+\Sigma_1|}{|\Sigma_1|}
			+\frac{1}{2}\log\frac{|\Sigma_U+\Sigma_V+\Sigma_g|}{|\Sigma_V+\Sigma_g|}.\label{eq:mimo-gauss-ach-r2}
		\end{align}
		On the other hand, since $\mathbf U$ is independent of $(\mathbf V,\mathbf G,\mathbf S)$ and $\mathbf Y_2-\mathbf U=\mathbf V+\mathbf G$, the conditional covariance formula for Gaussian random vectors yields
		\begin{align}\label{eq:mimo-gauss-ach-dist}
			d^*(\mathbf S|\mathbf U,\mathbf Y_2)
			&=\operatorname{tr}\operatorname{Cov}(\mathbf S|\mathbf V+\mathbf G)\notag\\
			&=\operatorname{tr}\left[
			\Sigma_s-\Sigma_s(\Sigma_V+\Sigma_g)^{-1}\Sigma_s
			\right]\notag\\
			&=\operatorname{tr}\Sigma_e(\Sigma_V),
		\end{align}
		where $\Sigma_e(A):=\Sigma_s-\Sigma_s(A+\Sigma_g)^{-1}\Sigma_s$.
		
		Therefore,
		\begin{equation}\label{eq:mimo-gauss-lb}
			C(D,P)\ge\underline{C}(D,P)\ge R_{\mathrm{G}}(D,P),
		\end{equation}
		where
		\begin{subequations}\label{eq:mimo-CG}
			\begin{align}
				R_{\mathrm{G}}(D,P):=&\max_{\Sigma_U,\Sigma_V\succeq0}\;
				\min\bigg\{\frac{1}{2}\log\frac{|\Sigma_U+\Sigma_V+\Sigma_1|}{|\Sigma_1|},\notag\\
				&\frac{1}{2}\log\frac{|\Sigma_V+\Sigma_1|}{|\Sigma_1|}+\frac{1}{2}\log\frac{|\Sigma_U+\Sigma_V+\Sigma_g|}{|\Sigma_V+\Sigma_g|}\bigg\}\label{eq:mimo-CG-a}\\
				\text{s.t.}\;\;
				&\operatorname{tr}\Sigma_e(\Sigma_V)\le D,\label{eq:mimo-CG-b}\\
				&\operatorname{tr}(\Sigma_U+\Sigma_V)\le P.\label{eq:mimo-CG-c}
			\end{align}
		\end{subequations}
		
		\subsubsection{Proof of the Upper Bound}\label{subsubsec:mimo-conv-proof}

		It now suffices to show that $C(D,P)\le R_{\mathrm{G}}(D,P)$. To this end, we first present three lemmas, which characterize the sensing distortion constraint, establish an extremal inequality for the difference of conditional differential entropies, and construct an enhanced channel, respectively.
		
		\begin{lemma}[Distortion via Fisher Information]\label{lem:mimo-fisher}
			Let $J_{\mathbf{u}}:=J(\mathbf{Y}_2|\mathbf{U}=\mathbf{u})$ denote the conditional Fisher information matrix, and define $\bar{J}:=\mathbb{E}_{\mathbf{U}}[J_{\mathbf{U}}]$ and the equivalent covariance matrix $\hat{\Sigma}_V:=\bar{J}^{-1}-\Sigma_g$. Then
			\begin{enumerate}[\IEEEsetlabelwidth{(ii)}]
				\item[(i)] $\mathbf{0}\preceq\hat{\Sigma}_V\preceq\Sigma_X$; and
				\item[(ii)] $d^*(\mathbf{S}|\mathbf{U},\mathbf{Y}_2)=\operatorname{tr}\Sigma_e(\hat{\Sigma}_V)$.
			\end{enumerate}
		\end{lemma}
		\begin{IEEEproof}
			See Appendix~\ref{app:lem-fisher}.
		\end{IEEEproof}
		
		\begin{lemma}[Extremal Inequality]\label{lem:mimo-extremal}
			Let $\tilde{\mathbf{Y}}_1:=\mathbf{X}+\tilde{\mathbf{Z}}_1$, where $\tilde{\mathbf{Z}}_1\sim\mathcal{N}(\mathbf{0},\tilde{\Sigma}_1)$ with $\mathbf{0}\prec\tilde{\Sigma}_1\preceq\Sigma_g$ is independent of $(\mathbf{U},\mathbf{X})$, and let $\hat{\Sigma}_V$ be as defined in Lemma~\ref{lem:mimo-fisher}. Then, for any $p_{\mathbf{U}\mathbf{X}}$,
			\begin{equation}\label{eq:mimo-entropy-bound}
				h(\tilde{\mathbf{Y}}_1|\mathbf{U})-h(\mathbf{Y}_2|\mathbf{U})
				\le\frac{1}{2}\log\frac{|\hat{\Sigma}_V+\tilde{\Sigma}_1|}{|\hat{\Sigma}_V+\Sigma_g|}.
			\end{equation}
		\end{lemma}
		\begin{IEEEproof}
			The proof follows the lines of~\cite[Sec.~III-B]{EkremUlukus2012Alt}; see also the extremal inequalities in~\cite{LiuViswanath2007IT} and~\cite[Theorem~5]{EkremUlukus2011MRWTC}. A proof is given in Appendix~\ref{app:lem-extremal} for completeness.
		\end{IEEEproof}
		
		\begin{lemma}[Enhanced Channel Construction]\label{lem:mimo-kkt}
			For any $\delta>0$, define the convex compact set
			\begin{equation*}
				\mathcal{V}_\delta:=\big\{\Sigma_V:\mathbf{0}\preceq\Sigma_V\preceq\Sigma_X+\delta\mathbf{I},\ \Sigma_V\preceq\hat{\Sigma}_V+\delta\mathbf{I}\big\},
			\end{equation*}
			and consider the optimization problem $\max_{\Sigma_V\in\mathcal{V}_\delta}f(\Sigma_V)$, where $f(\Sigma_V):=\log|\Sigma_V+\Sigma_1|-\log|\Sigma_V+\Sigma_g|$. Let $\Sigma_V^*$ be an optimal solution, and let $L_0\succeq\mathbf{0}$ be the Lagrange multiplier associated with the constraint $\Sigma_V\succeq\mathbf{0}$, which satisfies $L_0\Sigma_V^*=\mathbf{0}$. Define $\tilde{\Sigma}_1:=(\Sigma_1^{-1}+L_0)^{-1}$ and $\tilde{f}(\Sigma_V):=\log|\Sigma_V+\tilde{\Sigma}_1|-\log|\Sigma_V+\Sigma_g|$. Then
			\begin{enumerate}[\IEEEsetlabelwidth{(iii)}]
				\item[(i)] $\tilde{\Sigma}_1\preceq\Sigma_1$ and $\tilde{\Sigma}_1\preceq\Sigma_g$;
				\item[(ii)] $\dfrac{|\Sigma_V^*+\tilde{\Sigma}_1|}{|\tilde{\Sigma}_1|}=\dfrac{|\Sigma_V^*+\Sigma_1|}{|\Sigma_1|}$; and
				\item[(iii)] $\tilde{f}$ is concave on the positive semidefinite cone, and $\Sigma_V^*\in\arg\max_{\Sigma_V\in\mathcal{V}_\delta}\tilde{f}(\Sigma_V)$.
			\end{enumerate}
		\end{lemma}
		\begin{IEEEproof}
			See Appendix~\ref{app:lem-kkt}.
		\end{IEEEproof}
		
		Lemma~\ref{lem:mimo-kkt} extends the channel enhancement technique in~\cite{Weingarten2006IT} to the setting with a sensing distortion constraint: the constraint $\Sigma_V\preceq\hat{\Sigma}_V+\delta\mathbf{I}$ in the feasible set $\mathcal{V}_\delta$ stems from the distortion characterization in Lemma~\ref{lem:mimo-fisher}. For any $\delta>0$, let $\Sigma_V^*$ and $\tilde{\Sigma}_1$ be as defined in Lemma~\ref{lem:mimo-kkt}, and let $\Sigma_U^*:=\Sigma_X+\delta\mathbf{I}-\Sigma_V^*\succeq\mathbf{0}$.
		
		We now show that $C(D,P)\le R_{\mathrm{G}}(D,P)$.
			Fix an arbitrary $p_{\mathbf{U}\mathbf{X}}$ satisfying
		\eqref{eq:mimo-conv-bound-b} and
		\eqref{eq:mimo-conv-bound-c}, and let $\delta>0$.
		By Remark \ref{rem:mimo-coupling}, the admissible coupling may be chosen
		separately for this distribution. 
We adopt the coupling identified below using
Lemma~\ref{lem:mimo-kkt}.
		Construct the enhanced channel output
		\begin{equation}
			\tilde{\mathbf{Y}}_1=\mathbf{X}+\tilde{\mathbf{Z}}_1,\quad
			\tilde{\mathbf{Z}}_1\sim\mathcal{N}(\mathbf{0},\tilde{\Sigma}_1),
		\end{equation}
		such that both $\mathbf{Y}_1$ and $\mathbf{Y}_2$ are degraded versions of $\tilde{\mathbf{Y}}_1$, i.e.,
		\begin{align}
			\mathbf{Y}_1&=\tilde{\mathbf{Y}}_1+\mathbf{W}_1,\quad
			\mathbf{W}_1\sim\mathcal{N}(\mathbf{0},\Sigma_1-\tilde{\Sigma}_1),\\
			\mathbf{Y}_2&=\tilde{\mathbf{Y}}_1+\mathbf{W}_2,\quad
			\mathbf{W}_2\sim\mathcal{N}(\mathbf{0},\Sigma_g-\tilde{\Sigma}_1),
		\end{align}
		where $\tilde{\mathbf{Z}}_1$, $\mathbf{W}_1$, and $\mathbf{W}_2$ are mutually independent and independent of $(\mathbf{U},\mathbf{X})$. Since $\tilde{\mathbf{Z}}_1+\mathbf{W}_1\sim\mathcal{N}(\mathbf{0},\Sigma_1)$ and $\tilde{\mathbf{Z}}_1+\mathbf{W}_2\sim\mathcal{N}(\mathbf{0},\Sigma_g)$, this construction preserves $p_{\mathbf{Y}_1|\mathbf{X}}$ and $p_{\mathbf{Y}_2\mathbf{S}|\mathbf{X}}$, provided that the conditional distribution of $\mathbf{S}$ given $\tilde{\mathbf{Z}}_1+\mathbf{W}_2$ is the same as that of $\mathbf{S}$ given $\mathbf{G}$ in the original channel. 
		Under this coupling, the Markov chain $(\mathbf{U},\mathbf{X})\leftrightarrow\tilde{\mathbf{Y}}_1\leftrightarrow(\mathbf{Y}_1,\mathbf{Y}_2)$ holds by construction. We now bound the two terms in~\eqref{eq:mimo-conv-bound-a} for the fixed $p_{\mathbf{U}\mathbf{X}}$. Since the coupling may depend on $p_{\mathbf{U}\mathbf{X}}$, these pointwise bounds can be used in the strengthened converse described in Remark~\ref{rem:mimo-coupling}.
		 For the second term,
		\begin{align}
			&I(\mathbf{X};\mathbf{Y}_1,\mathbf{Y}_2|\mathbf{U})+I(\mathbf{U};\mathbf{Y}_2)\notag\\
			&\overset{(a)}{\le}I(\mathbf{X};\tilde{\mathbf{Y}}_1,\mathbf{Y}_2|\mathbf{U})+I(\mathbf{U};\mathbf{Y}_2)\notag\\
			&=I(\mathbf{X};\tilde{\mathbf{Y}}_1|\mathbf{U})+I(\mathbf{X};\mathbf{Y}_2|\mathbf{U},\tilde{\mathbf{Y}}_1)+I(\mathbf{U};\mathbf{Y}_2)\notag\\
			&\overset{(b)}{=}I(\mathbf{X};\tilde{\mathbf{Y}}_1|\mathbf{U})+I(\mathbf{U};\mathbf{Y}_2)\notag\\
			&=h(\tilde{\mathbf{Y}}_1|\mathbf{U})-h(\mathbf{Y}_2|\mathbf{U})-h(\tilde{\mathbf{Z}}_1)+h(\mathbf{Y}_2)\notag\\
			&\overset{(c)}{\le}\frac{1}{2}\log\frac{|\hat{\Sigma}_V+\tilde{\Sigma}_1|}{|\hat{\Sigma}_V+\Sigma_g|}-\frac{1}{2}\log|\tilde{\Sigma}_1|+\frac{1}{2}\log|\Sigma_X+\Sigma_g|\notag\\
			&=\frac{1}{2}\tilde{f}(\hat{\Sigma}_V)-\frac{1}{2}\log|\tilde{\Sigma}_1|+\frac{1}{2}\log|\Sigma_X+\Sigma_g|\notag\\
			&\overset{(d)}{\le}\frac{1}{2}\tilde{f}(\Sigma_V^*)-\frac{1}{2}\log|\tilde{\Sigma}_1|+\frac{1}{2}\log|\Sigma_X+\delta\mathbf{I}+\Sigma_g|\notag\\
			&\overset{(e)}{=}\frac{1}{2}\log\frac{|\Sigma_V^*+\Sigma_1|}{|\Sigma_1|}+\frac{1}{2}\log\frac{|\Sigma_U^*+\Sigma_V^*+\Sigma_g|}{|\Sigma_V^*+\Sigma_g|},\label{eq:mimo-R2-bound-final}
		\end{align}
		where $(a)$ follows from $(\mathbf{U},\mathbf{X})\leftrightarrow(\tilde{\mathbf{Y}}_1,\mathbf{Y}_2)\leftrightarrow\mathbf{Y}_1$; $(b)$ follows from $(\mathbf{U},\mathbf{X})\leftrightarrow\tilde{\mathbf{Y}}_1\leftrightarrow\mathbf{Y}_2$; $(c)$ follows from Lemma~\ref{lem:mimo-extremal}, $h(\tilde{\mathbf{Z}}_1)=\frac{1}{2}\log\big((2\pi e)^K|\tilde{\Sigma}_1|\big)$, and $h(\mathbf{Y}_2)\le\frac{1}{2}\log\big((2\pi e)^K|\Sigma_X+\Sigma_g|\big)$; $(d)$ follows from $\hat{\Sigma}_V\in\mathcal{V}_\delta$ (Lemma~\ref{lem:mimo-fisher}(i)), Lemma~\ref{lem:mimo-kkt}(iii), and the monotonicity of the log-determinant; and $(e)$ follows from Lemma~\ref{lem:mimo-kkt}(ii) and $\Sigma_U^*+\Sigma_V^*=\Sigma_X+\delta\mathbf{I}$.
		
		On the other hand, by the maximum entropy property of the Gaussian distribution and $\Sigma_X\preceq\Sigma_U^*+\Sigma_V^*$,
		\begin{align}\label{eq:mimo-R1-bound-final}
			I(\mathbf{X};\mathbf{Y}_1)
			&\le\frac{1}{2}\log\frac{|\Sigma_X+\Sigma_1|}{|\Sigma_1|}\notag\\
			&\le\frac{1}{2}\log\frac{|\Sigma_U^*+\Sigma_V^*+\Sigma_1|}{|\Sigma_1|}.
		\end{align}
		The right-hand sides of~\eqref{eq:mimo-R2-bound-final} and~\eqref{eq:mimo-R1-bound-final} are exactly the two terms in~\eqref{eq:mimo-CG-a} evaluated at the covariance matrices $(\Sigma_U^*,\Sigma_V^*)$.
		
		We next verify that $(\Sigma_U^*,\Sigma_V^*)$ approximately satisfies the power and distortion constraints in~\eqref{eq:mimo-CG}. First, by $\Sigma_U^*+\Sigma_V^*=\Sigma_X+\delta\mathbf{I}$ and~\eqref{eq:mimo-conv-bound-c}, we have
		\begin{align}
			\operatorname{tr}(\Sigma_U^*+\Sigma_V^*)
			&=\operatorname{tr}(\Sigma_X)+K\delta \notag\\
			&\le P+K\delta.
		\end{align}
		Second, let $\hat{\Sigma}_g := \hat{\Sigma}_V + \Sigma_g$. Since $\Sigma_e(\cdot)$ is monotonically increasing with respect to the positive semidefinite order, combining this with $\Sigma_V^*\preceq\hat{\Sigma}_V+\delta\mathbf{I}$ yields
		\begin{align}
			\operatorname{tr}\Sigma_e(\Sigma_V^*) 
			&\le \operatorname{tr}\Sigma_e(\hat{\Sigma}_V + \delta \mathbf{I}) \notag\\
			&= \operatorname{tr}\Sigma_e(\hat{\Sigma}_V) + \operatorname{tr}\Big[ \Sigma_s \big( \hat{\Sigma}_g^{-1} - (\hat{\Sigma}_g + \delta \mathbf{I})^{-1} \big) \Sigma_s \Big] \notag\\
			&\overset{(a)}{\le} \operatorname{tr}\Sigma_e(\hat{\Sigma}_V) + \delta \operatorname{tr}\big( \hat{\Sigma}_g^{-2} \Sigma_s^2 \big) \notag\\
			&\overset{(b)}{\le} \operatorname{tr}\Sigma_e(\hat{\Sigma}_V) + \underbrace{\|\Sigma_g^{-1}\|_2^2 \operatorname{tr}(\Sigma_s^2)}_{\kappa} \delta \notag\\
			&{\le} D + \kappa \delta,
		\end{align}
		where $(a)$ follows from the matrix inverse identity $\hat{\Sigma}_g^{-1} - (\hat{\Sigma}_g + \delta \mathbf{I})^{-1} = \delta \hat{\Sigma}_g^{-1}(\hat{\Sigma}_g + \delta \mathbf{I})^{-1} \preceq \delta \hat{\Sigma}_g^{-2}$, and $(b)$ follows from the trace inequality $\operatorname{tr}(\hat{\Sigma}_g^{-2}\Sigma_s^2) \le \|\hat{\Sigma}_g^{-2}\|_2\operatorname{tr}(\Sigma_s^2)$ for positive semidefinite matrices, together with $\|\hat{\Sigma}_g^{-2}\|_2 \le \|\Sigma_g^{-1}\|_2^2$, which is implied by $\hat{\Sigma}_g \succeq \Sigma_g \succ \mathbf{0}$.
		
		Therefore, $(\Sigma_U^*,\Sigma_V^*)$ is a feasible solution to problem~\eqref{eq:mimo-CG} under the constraints $(D+\kappa\delta,P+K\delta)$. 
		Hence, by ~\eqref{eq:mimo-R2-bound-final}, \eqref{eq:mimo-R1-bound-final}, and the strengthened converse in Remark~\ref{rem:mimo-coupling},
		  $C(D,P)\le R_{\mathrm{G}}(D+\kappa\delta,P+K\delta)$. Since the feasible set of~\eqref{eq:mimo-CG} is compact and both the objective and constraint functions are continuous, letting $\delta\to0^+$ yields $C(D,P)\le R_{\mathrm{G}}(D,P)$. Combining this with Section~\ref{subsubsec:mimo-ach-proof}, we have
		\begin{equation*}
			R_{\mathrm{G}}(D,P)\le\underline{C}(D,P)\le C(D,P)\le R_{\mathrm{G}}(D,P),
		\end{equation*}
		and hence $C(D,P)=R_{\mathrm{G}}(D,P)$.
		\begin{figure}[t]
			\centering
			\includegraphics[width=0.8\linewidth]{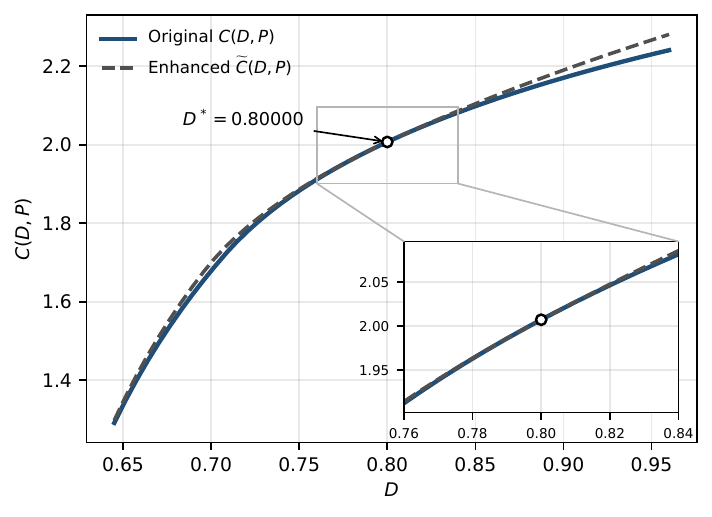}
			\caption{$C(D,P)$ of a $2\times2$ vector Gaussian ISAC channel and $\widetilde C(D,P)$ of the corresponding enhanced channel under the squared error distortion measure. The two capacity-distortion curves are tangent at $D^*=0.8$.}
			\label{fig:mimo-enhanced}
		\end{figure}
		
		Fig.~\ref{fig:mimo-enhanced} illustrates the above channel enhancement construction for $K=2$. Let the total power be $P=5$, and let
		\begin{equation*}
			\begin{gathered}
				\Sigma_1=
				\begin{pmatrix}
					0.378152&-0.519787\\
					-0.519787&1.916903
				\end{pmatrix},\\
				\Sigma_2=
				\begin{pmatrix}
					0.379263&0.150000\\
					0.150000&1.519805
				\end{pmatrix},\\
				\Sigma_s=
				\begin{pmatrix}
					1.788596&0.200000\\
					0.200000&0.424292
				\end{pmatrix}.
			\end{gathered}
		\end{equation*}
		In this case, $\Sigma_1$ and $\Sigma_g=\Sigma_2+\Sigma_s$ are not simultaneously orthogonally diagonalizable, and no positive semidefinite order exists between them. For a fixed distortion constraint $D^*=0.8$, the construction in Lemma~\ref{lem:mimo-kkt} yields
		\begin{equation*}
			\begin{aligned}
				L_0&\approx
				\begin{pmatrix}
					0.986512&-0.454823\\
					-0.454823&0.209692
				\end{pmatrix},\\
				\tilde\Sigma_1&\approx
				\begin{pmatrix}
					0.210647&-0.139235\\
					-0.139235&1.052334
				\end{pmatrix},
			\end{aligned}
		\end{equation*}
		where $\tilde\Sigma_1\preceq\Sigma_1$ and $\tilde\Sigma_1\preceq\Sigma_g$. Let $\widetilde C(D,P)$ denote the capacity obtained by replacing $\Sigma_1$ with $\tilde\Sigma_1$ in~\eqref{eq:mimo-CG}; then $\widetilde C(D,P)\geq C(D,P)$ for all feasible $D$. As shown in the figure, the two capacity-distortion curves are tangent at $(D^*,C(D^*,P))\approx(0.8,2.00716)$, and the capacity of the enhanced channel exceeds that of the original channel in a neighborhood on both sides of this point. This numerical result is consistent with the enhancement construction and the rate-preserving property in Lemma~\ref{lem:mimo-kkt}.

		\section{Conclusion}\label{sec:conclusion}
		
		This paper has studied the fundamental limits of bistatic ISAC systems over vector channels. For parallel channels in which the direction of degradation may differ across component channels, we have established a single-letter characterization of $C(D)$ and shown that jointly encoding the common layer across component channels exploits the complementary degradation directions of the two types of component channels, yielding a rate gain over independent coding on each component channel. For general vector channels without a global degradation order between the communication and sensing channels, we have derived single-letter lower and upper bounds on $C(D)$. In particular, for vector Gaussian channels with additive Gaussian states, we have shown that the lower bound is tight by means of a Fisher-information-based characterization of the distortion, the channel enhancement technique, and a matrix extremal inequality, thereby obtaining an exact characterization of $C(D,P)$: Gaussian superposition coding is optimal, and in the optimal scheme the sensing Rx performs partial decoding to assist state estimation.

		\appendices
		\section{Proof of Theorem~\ref{thm:mimo-conv}}\label{app:mimo-conv}
		
		By Fano's inequality,
		\begin{align}
			nR-n\epsilon_n
			&\le I(M;\mathbf{Y}_1^n)\notag\\
			&\le I(\mathbf{X}^n;\mathbf{Y}_1^n,\mathbf{Y}_2^n)\notag\\
			&=I(\mathbf{X}^n;\mathbf{Y}_2^n)+I(\mathbf{X}^n;\mathbf{Y}_1^n|\mathbf{Y}_2^n).\label{eq:mimo-conv-split}
		\end{align}
		For $t\in[1:n]$, let $\mathbf U_t:=(\mathbf Y_2^{t-1},\mathbf Y_{2,t+1}^n)$. We next bound the two terms in~\eqref{eq:mimo-conv-split} separately.
		
		For the first term,
		\begin{align}
			I(\mathbf{X}^n;\mathbf{Y}_2^n)
			&=\sum_{t=1}^{n}\big[H(\mathbf{Y}_{2,t}|\mathbf{Y}_2^{t-1})-H(\mathbf{Y}_{2,t}|\mathbf{X}^n,\mathbf{Y}_2^{t-1})\big]\notag\\
			&\overset{(a)}{=}\sum_{t=1}^{n}\big[H(\mathbf{Y}_{2,t}|\mathbf{Y}_2^{t-1})-H(\mathbf{Y}_{2,t}|\mathbf{X}_t)\big]\notag\\
			&{\le}\sum_{t=1}^{n}\big[H(\mathbf{Y}_{2,t})-H(\mathbf{Y}_{2,t}|\mathbf{X}_t,\mathbf{U}_t)\big]\notag\\
			&=\sum_{t=1}^{n}\big[I(\mathbf{U}_t;\mathbf{Y}_{2,t})+I(\mathbf{X}_t;\mathbf{Y}_{2,t}|\mathbf{U}_t)\big],\label{eq:mimo-conv-1}
		\end{align}
		where $(a)$ follows from the Markov chain $(\mathbf X^{n\setminus t},\mathbf U_t)\leftrightarrow\mathbf X_t\leftrightarrow(\mathbf Y_{1,t},\mathbf Y_{2,t},\mathbf S_t)$.
		
		For the second term,
		\begin{align}
			&I(\mathbf{X}^n;\mathbf{Y}_1^n|\mathbf{Y}_2^n)\notag\\
			&=\sum_{t=1}^{n}\big[H(\mathbf{Y}_{1,t}|\mathbf{Y}_1^{t-1},\mathbf{Y}_2^n)-H(\mathbf{Y}_{1,t}|\mathbf{X}^n,\mathbf{Y}_1^{t-1},\mathbf{Y}_2^n)\big]\notag\\
			&\overset{(a)}{=}\sum_{t=1}^{n}\big[H(\mathbf{Y}_{1,t}|\mathbf{Y}_1^{t-1},\mathbf{Y}_2^n)-H(\mathbf{Y}_{1,t}|\mathbf{X}_t,\mathbf{Y}_{2,t})\big]\notag\\
			&{\le}\sum_{t=1}^{n}\big[H(\mathbf{Y}_{1,t}|\mathbf{U}_t,\mathbf{Y}_{2,t})-H(\mathbf{Y}_{1,t}|\mathbf{X}_t,\mathbf{U}_t,\mathbf{Y}_{2,t})\big]\notag\\
			&=\sum_{t=1}^{n}I(\mathbf{X}_t;\mathbf{Y}_{1,t}|\mathbf{U}_t,\mathbf{Y}_{2,t}),\label{eq:mimo-conv-2}
		\end{align}
		where $(a)$ follows from the Markov chain $(\mathbf{X}^{n\setminus t},\mathbf{Y}_1^{t-1},\mathbf{U}_t)\leftrightarrow\mathbf{X}_t\leftrightarrow(\mathbf{Y}_{1,t},\mathbf{Y}_{2,t},\mathbf{S}_t)$.
		
		Let $Q\sim\mathrm{Unif}[1:n]$ be a time-sharing random variable independent of $(\mathbf{X}^n,\mathbf{S}^n,\mathbf{Y}_1^n,\mathbf{Y}_2^n)$, and define $\mathbf{U}:=(Q,\mathbf{U}_Q)$, $\mathbf{X}:=\mathbf{X}_Q$, $\mathbf{Y}_i:=\mathbf{Y}_{i,Q}$ ($i=1,2$), and $\mathbf{S}:=\mathbf{S}_Q$; then $\mathbf{U}\leftrightarrow\mathbf{X}\leftrightarrow(\mathbf{Y}_1,\mathbf{Y}_2,\mathbf{S})$. Substituting~\eqref{eq:mimo-conv-1} and~\eqref{eq:mimo-conv-2} into~\eqref{eq:mimo-conv-split}, dividing both sides by $n$, and applying the chain rule, we obtain
		\begin{align}
			R-\epsilon_n
			&\le\frac{1}{n}\sum_{t=1}^{n}\big[I(\mathbf{U}_t;\mathbf{Y}_{2,t})+I(\mathbf{X}_t;\mathbf{Y}_{1,t},\mathbf{Y}_{2,t}|\mathbf{U}_t)\big]\notag\\
			&=I(\mathbf{U}_Q;\mathbf{Y}_{2,Q}|Q)+I(\mathbf{X}_Q;\mathbf{Y}_{1,Q},\mathbf{Y}_{2,Q}|\mathbf{U}_Q,Q)\notag\\
			&\overset{(a)}{\le}I(\mathbf{U};\mathbf{Y}_2)+I(\mathbf{X};\mathbf{Y}_1,\mathbf{Y}_2|\mathbf{U}),\label{eq:mimo-conv-R2}
		\end{align}
		where $(a)$ follows from $I(\mathbf{U}_Q;\mathbf{Y}_{2,Q}|Q)\le I(Q,\mathbf{U}_Q;\mathbf{Y}_{2,Q})=I(\mathbf{U};\mathbf{Y}_2)$. Moreover, by~\eqref{eq:par-conv-fano},
		\begin{equation}\label{eq:mimo-conv-R1}
			R-\epsilon_n \le I(\mathbf{X};\mathbf{Y}_1).
		\end{equation}
		
		For the distortion constraint, we have
		\begin{align}\label{eq:mimo-conv-D}
			D+\epsilon
			&\ge \frac{1}{n}\sum_{t=1}^{n}\mathbb E\big[d(\mathbf S_t,\hat{\mathbf S}_t)\big]\notag\\
			&\ge\frac{1}{n}\sum_{t=1}^{n}d^*(\mathbf S_t|\mathbf Y_2^n)\notag\\
			&\overset{(a)}{=}\frac{1}{n}\sum_{t=1}^{n}d^*(\mathbf S_t|\mathbf U_t,\mathbf Y_{2,t})\notag\\
			&\overset{(b)}{=}d^*(\mathbf S|\mathbf U,\mathbf Y_2),
		\end{align}
		where $(a)$ follows from $\mathbf Y_2^n=(\mathbf U_t,\mathbf Y_{2,t})$ and $(b)$ follows from $\mathbf U=(Q,\mathbf U_Q)$.
		
		Combining~\eqref{eq:mimo-conv-R2}--\eqref{eq:mimo-conv-D}, we have $R-\epsilon_n\le\overline{C}(D+\epsilon)$. Letting $n\to\infty$ and $\epsilon\to0$ yields $C(D)\le\overline{C}(D)$.

		\section{Proof of Lemma~\ref{lem:mimo-fisher}}\label{app:lem-fisher}
		
		We first prove property (i). Let $\Sigma_{\mathbf{X}|\mathbf{u}}:=\operatorname{Cov}(\mathbf{X}|\mathbf{U}=\mathbf{u})$. Since $\mathbf{Y}_2=\mathbf{X}+\mathbf{G}$ and $\mathbf{G}\sim\mathcal{N}(\mathbf{0},\Sigma_g)$ is independent of $(\mathbf{U},\mathbf{X})$, we have
		\begin{align}
			\Sigma_g^{-1}
			&\overset{(a)}{\succeq}\bar{J}
			\overset{(b)}{\succeq}\mathbb{E}_{\mathbf{U}}\big[(\Sigma_{\mathbf{X}|\mathbf{U}}+\Sigma_g)^{-1}\big]\notag\\
			&\overset{(c)}{\succeq}\big(\mathbb{E}_{\mathbf{U}}[\Sigma_{\mathbf{X}|\mathbf{U}}]+\Sigma_g\big)^{-1}\notag\\
			&\overset{(d)}{\succeq}(\Sigma_X+\Sigma_g)^{-1},\label{eq:fisher-matrix-bound}
		\end{align}
		where $(a)$ follows from the Fisher information inequality~\cite{Zamir1998IT}, i.e., $J_{\mathbf{u}}=J(\mathbf{X}+\mathbf{G}|\mathbf{U}=\mathbf{u})\preceq J(\mathbf{G})=\Sigma_g^{-1}$ for any $\mathbf{u}$; $(b)$ follows from the Cram\'er--Rao-type inequality $J_{\mathbf{u}}\succeq\operatorname{Cov}(\mathbf{Y}_2|\mathbf{U}=\mathbf{u})^{-1}=(\Sigma_{\mathbf{X}|\mathbf{u}}+\Sigma_g)^{-1}$~\cite{DemboCoverThomas1991}; $(c)$ follows from the operator convexity of matrix inversion~\cite{Bhatia1997} and Jensen's inequality; and $(d)$ follows from $\mathbb{E}_{\mathbf{U}}[\Sigma_{\mathbf{X}|\mathbf{U}}]\preceq\operatorname{Cov}(\mathbf{X})\preceq\Sigma_X$ and the order-reversing property of matrix inversion. By~\eqref{eq:fisher-matrix-bound}, $\bar{J}\succ\mathbf{0}$; applying the order-reversing property again yields
		\begin{equation*}
			\mathbf{0}\preceq\bar{J}^{-1}-\Sigma_g=\hat{\Sigma}_V\preceq\Sigma_X,
		\end{equation*}
		which is property (i).
		
		We next prove property (ii). Given $\mathbf{U}=\mathbf{u}$, we have $\mathbf{Y}_2=\mathbf{S}+(\mathbf{X}+\mathbf{Z}_2)$, where $\mathbf{X}+\mathbf{Z}_2$ and $\mathbf{S}\sim\mathcal{N}(\mathbf{0},\Sigma_s)$ are conditionally independent. Hence,
		\begin{equation}
			p_{\mathbf{Y}_2|\mathbf{U}}(\mathbf{y}_2|\mathbf{u})
			=\int\phi_{\Sigma_s}(\mathbf{y}_2-\mathbf{w})\,\mathrm{d}P_{\mathbf{X}+\mathbf{Z}_2|\mathbf{U}=\mathbf{u}}(\mathbf{w}),
		\end{equation}
		where $\phi_{\Sigma_s}$ denotes the zero-mean Gaussian probability density function with covariance matrix $\Sigma_s$. Taking the gradient with respect to $\mathbf{y}_2$ yields
		\begin{align}
			&\nabla_{\mathbf{y}_2}\ln p_{\mathbf{Y}_2|\mathbf{U}}(\mathbf{y}_2|\mathbf{u})\notag\\
			&\quad=-\Sigma_s^{-1}\mathbb{E}\big[\mathbf{Y}_2-(\mathbf{X}+\mathbf{Z}_2)\,\big|\,\mathbf{Y}_2=\mathbf{y}_2,\mathbf{U}=\mathbf{u}\big]\notag\\
			&\quad=-\Sigma_s^{-1}\mathbb{E}\big[\mathbf{S}\,\big|\,\mathbf{Y}_2=\mathbf{y}_2,\mathbf{U}=\mathbf{u}\big].\label{eq:score-tweedie}
		\end{align}
		By the definition of the Fisher information matrix and~\eqref{eq:score-tweedie},
		\begin{equation}\label{eq:fisher-outer}
			\Sigma_sJ_{\mathbf{u}}\Sigma_s
			=\mathbb{E}\Big[\mathbb{E}[\mathbf{S}|\mathbf{Y}_2,\mathbf{U}=\mathbf{u}]\,\mathbb{E}[\mathbf{S}|\mathbf{Y}_2,\mathbf{U}=\mathbf{u}]^\top\,\Big|\,\mathbf{U}=\mathbf{u}\Big].
		\end{equation}
		On the other hand, since $\mathbf{S}$ is independent of $\mathbf{U}$, we have $\mathbb{E}[\mathbf{S}|\mathbf{U}=\mathbf{u}]=\mathbf{0}$ and $\operatorname{Cov}(\mathbf{S}|\mathbf{U}=\mathbf{u})=\Sigma_s$. By the law of total covariance,
		\begin{align}
			\Sigma_s
			&=\operatorname{Cov}(\mathbf{S}|\mathbf{Y}_2,\mathbf{U}=\mathbf{u})\notag\\
			&\quad+\mathbb{E}\Big[\mathbb{E}[\mathbf{S}|\mathbf{Y}_2,\mathbf{U}=\mathbf{u}]\,\mathbb{E}[\mathbf{S}|\mathbf{Y}_2,\mathbf{U}=\mathbf{u}]^\top\,\Big|\,\mathbf{U}=\mathbf{u}\Big].\label{eq:total-cov}
		\end{align}
		Combining~\eqref{eq:fisher-outer} and~\eqref{eq:total-cov} and taking the expectation over $\mathbf{U}$, we obtain
		\begin{align}\label{eq:fisher-distortion-identity}
			\operatorname{Cov}(\mathbf{S}|\mathbf{U},\mathbf{Y}_2)
			&=\Sigma_s-\Sigma_s\bar{J}\Sigma_s\notag\\
			&=\Sigma_s-\Sigma_s(\hat{\Sigma}_V+\Sigma_g)^{-1}\Sigma_s\notag\\
			&=\Sigma_e(\hat{\Sigma}_V),
		\end{align}
		where the second equality uses $\bar{J}=(\hat{\Sigma}_V+\Sigma_g)^{-1}$. By~\eqref{eq:mimo-var}, $d^*(\mathbf{S}|\mathbf{U},\mathbf{Y}_2)=\operatorname{tr}\Sigma_e(\hat{\Sigma}_V)$, which is property (ii).

		\section{Proof of Lemma~\ref{lem:mimo-extremal}}\label{app:lem-extremal}
		
		Let $\Delta:=\Sigma_g-\tilde{\Sigma}_1\succeq\mathbf{0}$. For $t\in[0,1]$, define the interpolated covariance matrix $\Omega_t:=\tilde{\Sigma}_1+t\Delta$ and the observation $\mathbf{Y}(t):=\mathbf{X}+\mathbf{N}(t)$, where $\mathbf{N}(t)\sim\mathcal{N}(\mathbf{0},\Omega_t)$ is independent of $(\mathbf{U},\mathbf{X})$. Then $(\mathbf{U},\mathbf{Y}(0))$ has the same distribution as $(\mathbf{U},\tilde{\mathbf{Y}}_1)$, and $(\mathbf{U},\mathbf{Y}(1))$ has the same distribution as $(\mathbf{U},\mathbf{Y}_2)$. Let $J_{\mathbf{u}}(t):=J(\mathbf{Y}(t)|\mathbf{U}=\mathbf{u})$ and $\bar{J}(t):=\mathbb{E}_{\mathbf{U}}[J_{\mathbf{U}}(t)]$; then $\bar{J}(1)=\bar{J}$.
		
		Let $\Gamma_t:=\Sigma_g-\Omega_t=(1-t)\Delta\succeq\mathbf{0}$. Since $\mathbf{Y}(1)$ has the same distribution as $\mathbf{Y}(t)$ plus an independent $\mathcal{N}(\mathbf{0},\Gamma_t)$ noise, the Fisher information inequality~\cite{Zamir1998IT} implies that $J_{\mathbf{u}}(1)^{-1}\succeq J_{\mathbf{u}}(t)^{-1}+\Gamma_t$, i.e.,
		\begin{equation}\label{eq:fii-step}
			J_{\mathbf{u}}(t)\succeq\big(J_{\mathbf{u}}(1)^{-1}-\Gamma_t\big)^{-1}.
		\end{equation}
		Define the mapping $\varphi_t(A):=(A^{-1}-\Gamma_t)^{-1}$. When $\Gamma_t\succ\mathbf{0}$, the identity
		\begin{equation}
			\varphi_t(A)=-\Gamma_t^{-1}+\Gamma_t^{-1}(\Gamma_t^{-1}-A)^{-1}\Gamma_t^{-1}
		\end{equation}
		and the operator convexity of matrix inversion~\cite{Bhatia1997} imply that $\varphi_t$ is operator convex on its domain; when $\Gamma_t$ is singular, the conclusion follows by continuity. Therefore,
		\begin{align}
			\bar{J}(t)
			&\overset{(a)}{\succeq}\mathbb{E}_{\mathbf{U}}\big[\varphi_t(J_{\mathbf{U}}(1))\big]
			\overset{(b)}{\succeq}\varphi_t(\bar{J})
			\overset{(c)}{=}(\hat{\Sigma}_V+\Omega_t)^{-1},\label{eq:j-bar-lower}
		\end{align}
		where $(a)$ follows from~\eqref{eq:fii-step}, $(b)$ follows from the operator Jensen inequality, and $(c)$ follows from $\hat{\Sigma}_V=\bar{J}^{-1}-\Sigma_g$.
		
		On the other hand,
		\begin{align}
			\frac{\mathrm{d}}{\mathrm{d}t}h(\mathbf{Y}(t)|\mathbf{U})
			&\overset{(a)}{=}\frac{1}{2\ln2}\operatorname{tr}\big(\bar{J}(t)\Delta\big)\notag\\
			&\overset{(b)}{\ge}\frac{1}{2\ln2}\operatorname{tr}\big((\hat{\Sigma}_V+\Omega_t)^{-1}\Delta\big)\notag\\
			&\overset{(c)}{=}\frac{1}{2}\frac{\mathrm{d}}{\mathrm{d}t}\log|\hat{\Sigma}_V+\Omega_t|,\label{eq:debruijn-bound}
		\end{align}
		where $(a)$ follows from the multivariate de Bruijn identity~\cite{PalomarVerdu2006IT} and $\dot{\Omega}_t=\Delta$, $(b)$ follows from~\eqref{eq:j-bar-lower} and $\Delta\succeq\mathbf{0}$, and $(c)$ follows from Jacobi's formula for the derivative of a determinant. Integrating~\eqref{eq:debruijn-bound} over $[0,1]$ yields
		\begin{equation}
			h(\mathbf{Y}_2|\mathbf{U})-h(\tilde{\mathbf{Y}}_1|\mathbf{U})
			\ge\frac{1}{2}\log\frac{|\hat{\Sigma}_V+\Sigma_g|}{|\hat{\Sigma}_V+\tilde{\Sigma}_1|},
		\end{equation}
		which, after rearrangement, gives~\eqref{eq:mimo-entropy-bound}.

		\section{Proof of Lemma~\ref{lem:mimo-kkt}}\label{app:lem-kkt}
		
		Since $\mathcal{V}_\delta$ is compact and $f$ is continuous, an optimal solution $\Sigma_V^*$ exists. Let $\Sigma_{X,\delta}:=\Sigma_X+\delta\mathbf{I}$ and $\hat{\Sigma}_{V,\delta}:=\hat{\Sigma}_V+\delta\mathbf{I}$. Since $\delta>0$, $\Sigma_V=\frac{\delta}{2}\mathbf{I}$ is a strictly feasible point, i.e., Slater's condition holds; hence, the KKT conditions are necessary for optimality. To simplify matrix differentiation, we multiply the objective function by the constant $\ln2$, and the corresponding Lagrangian is
		\begin{align}
			\mathcal{L}&=(\ln2)f(\Sigma_V)+\operatorname{tr}(L_0\Sigma_V)\notag\\
			&\quad-\operatorname{tr}\big[L_X(\Sigma_V-\Sigma_{X,\delta})\big]-\operatorname{tr}\big[L_V(\Sigma_V-\hat{\Sigma}_{V,\delta})\big],
		\end{align}
		where $L_0,L_X,L_V\succeq\mathbf{0}$ are Lagrange multipliers. Then $\Sigma_V^*$ satisfies
		\begin{align}
			\mathbf{0}&=(\Sigma_V^*+\Sigma_1)^{-1}-(\Sigma_V^*+\Sigma_g)^{-1}+L_0-L_X-L_V,\label{eq:kkt-stat}\\
			\mathbf{0}&=L_0\Sigma_V^*,\label{eq:kkt-slack-0}\\
			\mathbf{0}&=L_X(\Sigma_{X,\delta}-\Sigma_V^*),\label{eq:kkt-slack-X}\\
			\mathbf{0}&=L_V(\hat{\Sigma}_{V,\delta}-\Sigma_V^*).\label{eq:kkt-slack-XV}
		\end{align}
		
		\textit{Property (i):} Since $\tilde{\Sigma}_1^{-1}=\Sigma_1^{-1}+L_0\succeq\Sigma_1^{-1}$, the order-reversing property of matrix inversion yields $\tilde{\Sigma}_1\preceq\Sigma_1$. To show that $\tilde{\Sigma}_1\preceq\Sigma_g$, note that $\tilde{\Sigma}_1^{-1}=\Sigma_1^{-1}+L_0$ implies $\tilde{\Sigma}_1-\Sigma_1=-\tilde{\Sigma}_1L_0\Sigma_1$. Thus,
		\begin{align}
			&(\Sigma_V^*+\tilde{\Sigma}_1)\big[(\Sigma_V^*+\Sigma_1)^{-1}+L_0\big]\notag\\
			&\overset{(a)}{=}\mathbf{I}+(\tilde{\Sigma}_1-\Sigma_1)(\Sigma_V^*+\Sigma_1)^{-1}+\tilde{\Sigma}_1L_0\notag\\
			&=\mathbf{I}+\tilde{\Sigma}_1L_0\big[\mathbf{I}-\Sigma_1(\Sigma_V^*+\Sigma_1)^{-1}\big]\notag\\
			&=\mathbf{I}+\tilde{\Sigma}_1L_0\Sigma_V^*(\Sigma_V^*+\Sigma_1)^{-1}\notag\\
			&\overset{(b)}{=}\mathbf{I},
		\end{align}
		where both $(a)$ and $(b)$ follow from~\eqref{eq:kkt-slack-0}. Therefore,
		\begin{equation}\label{eq:inv-identity}
			(\Sigma_V^*+\tilde{\Sigma}_1)^{-1}=(\Sigma_V^*+\Sigma_1)^{-1}+L_0.
		\end{equation}
		Substituting~\eqref{eq:inv-identity} into~\eqref{eq:kkt-stat} yields
		\begin{align}
			(\Sigma_V^*+\tilde{\Sigma}_1)^{-1}&=(\Sigma_V^*+\Sigma_g)^{-1}+L_X+L_V\notag\\
			&\succeq(\Sigma_V^*+\Sigma_g)^{-1},\label{eq:enhanced-gradient}
		\end{align}
		and the order-reversing property then gives $\tilde{\Sigma}_1\preceq\Sigma_g$.
		
		\textit{Property (ii):} By $\tilde{\Sigma}_1^{-1}=\Sigma_1^{-1}+L_0$ and~\eqref{eq:kkt-slack-0},
		\begin{align}
			\frac{|\Sigma_V^*+\tilde{\Sigma}_1|}{|\tilde{\Sigma}_1|}
			&=\big|\mathbf{I}+(\Sigma_1^{-1}+L_0)\Sigma_V^*\big|\notag\\
			&=\big|\mathbf{I}+\Sigma_1^{-1}\Sigma_V^*\big|\notag\\
			&=\frac{|\Sigma_V^*+\Sigma_1|}{|\Sigma_1|}.
		\end{align}
		
		\textit{Property (iii):} Let $\Delta:=\Sigma_g-\tilde{\Sigma}_1$, which satisfies $\Delta\succeq\mathbf{0}$ by property (i). Then $\tilde{f}(\Sigma_V)=-g(\Sigma_V+\tilde{\Sigma}_1)$, where
		\begin{equation}
			g(\Sigma_N):=\log|\Sigma_N+\Delta|-\log|\Sigma_N|,\quad\Sigma_N\succ\mathbf{0}.
		\end{equation}
		We now show that $g$ is convex. For any $\Sigma_{N,1},\Sigma_{N,2}\succ\mathbf{0}$ and $\lambda\in[0,1]$, let $\Sigma_{N,\lambda}:=\lambda\Sigma_{N,1}+(1-\lambda)\Sigma_{N,2}$. Let $\mathbf{B}\sim\mathcal{N}(\mathbf{0},\Delta)$ and $T\sim\operatorname{Bern}(\lambda)$ be mutually independent, and let $\mathbf{N}$ be distributed as $\mathcal{N}(\mathbf{0},\Sigma_{N,1})$ if $T=1$ and as $\mathcal{N}(\mathbf{0},\Sigma_{N,2})$ if $T=0$, independently of $\mathbf{B}$. Then $\operatorname{Cov}(\mathbf{N})=\Sigma_{N,\lambda}$, and hence
		\begin{align}
			g(\Sigma_{N,\lambda})
			&\overset{(a)}{\le}2I(\mathbf{B};\mathbf{B}+\mathbf{N})\notag\\
			&\overset{(b)}{\le}2I(\mathbf{B};\mathbf{B}+\mathbf{N},T)\notag\\
			&\overset{(c)}{=}2I(\mathbf{B};\mathbf{B}+\mathbf{N}|T)\notag\\
			&\overset{(d)}{=}\lambda g(\Sigma_{N,1})+(1-\lambda)g(\Sigma_{N,2}),
		\end{align}
		where $(a)$ follows from the worst additive noise theorem~\cite{DiggaviCover2001IT}, $(b)$ follows from the chain rule and the nonnegativity of mutual information, $(c)$ follows from the independence of $\mathbf{B}$ and $T$, and $(d)$ holds since $\mathbf{N}$ is Gaussian given $T$. Therefore, $g$ is convex, and $\tilde{f}$ is concave on the positive semidefinite cone.
		
		Finally, by~\eqref{eq:enhanced-gradient}, the gradient of $\tilde{f}$ at $\Sigma_V^*$ is $\frac{1}{\ln2}\big[(\Sigma_V^*+\tilde{\Sigma}_1)^{-1}-(\Sigma_V^*+\Sigma_g)^{-1}\big]=\frac{1}{\ln2}(L_X+L_V)$. By the concavity of $\tilde{f}$, for any $\Sigma_V\in\mathcal{V}_\delta$,
		\begin{align}
			\tilde{f}(\Sigma_V)
			&\le\tilde{f}(\Sigma_V^*)+\frac{1}{\ln2}\operatorname{tr}\big[(L_X+L_V)(\Sigma_V-\Sigma_V^*)\big]\notag\\
			&\overset{(a)}{=}\tilde{f}(\Sigma_V^*)+\frac{1}{\ln2}\operatorname{tr}\big[L_X(\Sigma_V-\Sigma_{X,\delta})\big]\notag\\
			&\qquad+\frac{1}{\ln2}\operatorname{tr}\big[L_V(\Sigma_V-\hat{\Sigma}_{V,\delta})\big]\notag\\
			&\overset{(b)}{\le}\tilde{f}(\Sigma_V^*),
		\end{align}
		where $(a)$ follows from~\eqref{eq:kkt-slack-X}–\eqref{eq:kkt-slack-XV} and $(b)$ follows from $L_X,L_V\succeq\mathbf{0}$, $\Sigma_V\preceq\Sigma_{X,\delta}$, and $\Sigma_V\preceq\hat{\Sigma}_{V,\delta}$. Therefore, $\Sigma_V^*\in\arg\max_{\Sigma_V\in\mathcal{V}_\delta}\tilde{f}(\Sigma_V)$, which establishes property (iii).

	\end{document}